\documentclass[sigconf]{acmart}

\makeatletter
\def\@copyrightspace{\relax}
\makeatother
\renewcommand\footnotetextcopyrightpermission[1]{}

\AtBeginDocument{%
  \providecommand\BibTeX{{%
    \normalfont B\kern-0.5em{\scshape i\kern-0.25em b}\kern-0.8em\TeX}}}

\acmBooktitle{arXiv Preprint 2021}
\acmConference[arXiv Preprint]{arXiv Preprint}{2021}{}
\acmPrice{15.00}
\acmISBN{978-1-4503-XXXX-X/18/06}
\usepackage{stmaryrd}
\newcommand{\ldqm}[1]{ \pmb{\left\llbracket\vphantom{#1}\right.}  #1   \pmb{\left.\vphantom{#1}\right\rrbracket} }

\usepackage{booktabs}

\usepackage[inline]{enumitem}

\usepackage{float}
\usepackage{listings}

\newfloat{lstfloat}{htbp}{lop}
\floatname{lstfloat}{Listing}
\usepackage[ruled,vlined]{algorithm2e}
\usepackage[capitalise]{cleveref}
\Crefname{definition}{Def.}{Defs.}

\usepackage{siunitx}

\usepackage{hyperref}

\usepackage{rotating}

\usepackage{amsmath}
\DeclareMathOperator*{\argmin}{arg\,min}
\DeclareMathOperator*{\argmax}{arg\,max}

\usepackage{cases}

\usepackage{subfig}

\usepackage{graphicx}
\usepackage{tikz}

\usepackage{booktabs}
\usepackage{multirow}

\DeclareFontEncoding{LS1}{}{}
\DeclareFontSubstitution{LS1}{stix}{m}{n}
\DeclareSymbolFont{STIXsymbols}{LS1}{stixscr}{m}{n}
\DeclareMathSymbol{\Bowtie}{\mathrel}{STIXsymbols}{"0E}

\newcommand{\fedi}{\textsc{Fed-I}\xspace}
\newcommand{\fedii}{\textsc{Fed-II}\xspace}

\begin{document}

\title{A Framework for Federated SPARQL Query Processing over Heterogeneous Linked Data Fragments}

%\author{}
\author{Lars Heling}
\email{lars.heling@kit.edu}
\affiliation{%
  \institution{Karlsruhe Institute of Technology}
  \streetaddress{Kaiserstr. 12}
  \city{Karlsruhe}
  \state{Germany}
}

\author{Maribel Acosta}
\email{maribel.acosta@rub.de}
\affiliation{%
  \institution{Ruhr-University Bochum}
  \streetaddress{Universitätsstraße 150}
  \city{Bochum}
  \state{Germany}
}

\renewcommand{\shortauthors}{Heling and Acosta}

\begin{abstract}
Linked Data Fragments (LDFs) refer to Web interfaces that allow for accessing and querying Knowledge Graphs on the Web.
These interfaces, such as SPARQL endpoints or Triple Pattern Fragment servers, differ in the SPARQL expressions they can evaluate and the metadata they provide.
Client-side query processing approaches have been proposed and are tailored to evaluate queries over individual interfaces.
Moreover, federated query processing has focused on federations with a single type of LDF interface, typically SPARQL endpoints.
In this work, we address the challenges of SPARQL query processing over federations with heterogeneous LDF interfaces.
To this end, we formalize the concept of federations of Linked Data Fragment and propose a framework for federated querying over heterogeneous federations with different LDF interfaces.
The framework comprises query decomposition, query planning, and physical operators adapted to the particularities of different LDF interfaces.
Further, we propose an approach for each component of our framework and evaluate them in an experimental study on the well-known FedBench benchmark.
The results show a substantial improvement in performance achieved by devising these interface-aware approaches exploiting the capabilities of heterogeneous interfaces in federations.

\end{abstract}

\maketitle

\section{Introduction}
\label{sec:introduction}
% Context
The increasing number and size of Knowledge Graphs published as Linked Data led to the development of different interfaces to support querying Knowledge Graphs on the Web~\cite{DBLP:journals/ws/VerborghSHHVMHC16,DBLP:journals/corr/HartigA16,DBLP:conf/www/MinierSM19,DBLP:conf/www/AzzamFABP20}.
These interfaces mainly differ in their expressivity, server availability, and client cost as shown in \Cref{fig:ldf_spectrum}.
The Linked Data Fragment (LDF) framework provides a uniform way to describe these interfaces regarding their querying expressivity and the metadata they provide~\cite{DBLP:journals/ws/VerborghSHHVMHC16,DBLP:conf/semweb/HartigLP17}.
The query expressivity of these interfaces ranges from triple patterns in Triple Pattern Fragment (TPF) servers to the full fragment of SPARQL in SPARQL endpoints.
%
%The Linked Data Fragment (LDF) framework provides a uniform way to define and explore different types of such Web  interfaces~\cite{DBLP:journals/ws/VerborghSHHVMHC16,DBLP:conf/semweb/HartigLP17}.
To support efficient querying, these developments also drove the research in the area of client-side SPARQL query processing tailored to the individual interfaces~\cite{DBLP:conf/semweb/TaelmanHSV18,DBLP:conf/semweb/AcostaV15,DBLP:conf/semweb/MontoyaAH18a}.
%
% Challenges
Therefore, most of the existing approaches focus on querying data from a single dataset or through a federation of sources with the same interfaces. 
However, the problem of evaluating SPARQL queries over heterogeneous federations of such LDF interfaces has not gained much attention thus far. 
To devise efficient querying plans in this scenario, it is not sufficient to combine and fuse existing solutions because the capabilities and limitations of the interfaces have to be taken into account altogether. 
An effective solution, thus, requires to re-define the notions of the main tasks of federated engines -- i.e query decomposition, planning, and execution -- to integrate the different capabilities of the interfaces in a single query plan.    
% Solution at a glance

\begin{figure}[t]
    \centering
    %%\vspace{4mm}
    \includegraphics[width=0.95\linewidth]{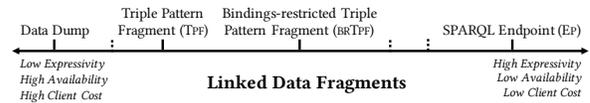}
    \caption{Linked Data Fragment Spectrum (based on \cite{DBLP:journals/ws/VerborghSHHVMHC16}).}
    \label{fig:ldf_spectrum}
    %\vspace{-7mm}
\end{figure}

In this work, we formalize federations of LDF interfaces and propose a framework that serves as a  foundation for devising efficient solutions for querying heterogeneous federations.  
At the core of the framework, we present novel concepts for query decomposition, planning, and execution in heterogeneous federations. 
These components include the definitions of (i) interface-compliant sub-expressions and the semantics of their evaluation, (ii) interface-aware query planning, and (iii) polymorphic physical operators that implement different execution strategies according to the contacted interfaces.  
To accompany our theoretical contributions, we propose simple yet novel approaches to query heterogeneous LDF federations.
Each approach addresses the particularities of the interfaces and is designed to reduce the query execution times and the load on the members of the federations. 
Our results show the effectiveness of our framework and illustrate how leveraging the interfaces' capabilities in a single plan can substantially improve query execution. 
%and therefore, shows the potential for further research in this area.
%
% Scope of the paper
In summary, the contributions of this work are
\begin{itemize}[leftmargin=.8\parindent]
    \item a general definition of Linked Data Fragment (LDF) federations,
    \item a framework for querying heterogeneous LDF federations addressing query decomposition, planning, and physical operators,
    \item a practical solution for each component of the framework, and 
    \item an experimental evaluation of a prototypical implementation of the solutions on heterogeneous federations.% with three common LDF interfaces.
\end{itemize}
%
% Structure of the paper
The remainder of this work is structured as follows. 
\cref{sec:motivating_example} presents a  motivating example, and in ~\cref{sec:ldf_federations}, we present our definition of federations of LDFs. 
Our framework is presented in \cref{sec:approach} and evaluated in \cref{sec:experimental_evaluation}.
We discuss related work in \cref{sec:related_work} and conclude our work in \cref{sec:conclusion}.
\vspace{3.5em}

% Reasons for heterogenous federations:
% - Different data providers choose only a specific interface to publish their data
% - Potentially, services are down, such that an alternative interfaces needs to be chosen

% comunica-sparql handels: [https://github.com/comunica/comunica/tree/master/packages/actor-init-sparql#readme]
% TPF, SPARQL, local HDT, RDFJS

%The increasing number and size of Knowledge Graphs published as Linked Data led to the development of different interfaces to support querying these Knowledge Graphs on the Web, ranging from Triple Pattern Fragment servers to SPARQL endpoints. These interfaces are mainly characterized by their emphasis on querying availability and expressivity and the Linked Data Fragment Framework provides a uniform way to describe these interfaces regarding their querying capabilities and the metadata they provide.

\section{Motivating Example}
\label{sec:motivating_example}
As a motivating example, consider the query to \textit{get American Presidents, the political party they are a member of as well as their predecessors and successors} shown in \cref{lst:query}.
\vspace{-0.5\baselineskip}
%\begin{lstfloat}[b!]
\begin{lstlisting}[
%float=b,
basicstyle=\sffamily\footnotesize,
language=SPARQL, 
mathescape=true,
morekeywords={}, commentstyle=\color{gray},
firstnumber=0, 
%numbers=left,
numbersep=5pt,
numberstyle=\scriptsize,
frame=tb,
%caption = {A Query to get American presidents, their political party, predecessors, and successors. Prefixes as in \url{http://prefix.cc}.}, 
caption = {Example query. Prefixes as in \url{http://prefix.cc}},
 label={lst:query},
 escapechar=",
columns=flexible,
numberfirstline=false
]
SELECT * WHERE {
    ?x wdt:P39 wd:Q11696 .                  # $tp_1$
    ?x wdt:P102 ?party .                    # $tp_2$
    ?y owl:sameAs ?x .                      # $tp_3$
    ?y dbo:predecessor ?predecessor .       # $tp_4$
    ?y dbo:successor ?successor .   }       # $tp_5$
\end{lstlisting}
%\end{lstfloat}
Let us assume, we want to evaluate the query over a federation that consists of the SPARQL endpoint of Wikidata\footnote{https://query.wikidata.org/sparql\label{fn:wd_sparql}} and the Triple Pattern Fragment (TPF) server \cite{DBLP:journals/ws/VerborghSHHVMHC16} of DBpedia\footnote{\url{http://fragments.dbpedia.org/2016-04/en}\label{fn:dbp_tpf}}.
The Wikidata endpoint provides solutions to triple patterns $tp_1, tp_2, tp_3$ and the DBpedia TPF server to $tp_3$, $tp_4$, and $tp_5$.
As the members of the federation implement different Linked Data Fragment (LDF) interfaces, we denote such a federation \textit{heterogeneous}.
In this example, we are not able to apply an existing query decomposition approach from query processing over SPARQL endpoint federations, as these approaches do not consider the capabilities of the LDF interfaces.
For instance, FedX \cite{DBLP:conf/semweb/SchwarteHHSS11} would group triple patterns $tp_4$ and $tp_5$ into a subquery which is not \emph{compliant} with the DBpedia TPF interface.
On the contrary, a naive decomposition that evaluates the query on the triple pattern level at the relevant sources would be \emph{compliant} with the interfaces in the federation, since they are all able to evaluate triple patterns.
However, such a decomposition leads to inefficient query plans on SPARQL endpoints as they produce an excessive number of requests on the server and thus, lead to long execution times.
In our example, this approach would fail to evaluate the subexpression $(tp_1\; \textsc{And}\; tp_2 )$ at the Wikidata endpoint.
Their individual evaluation at the endpoint leads to an overhead in requests and intermediate results transferred that could be avoided.
The example illustrates the challenges that arise in heterogeneous federations and motivates our research to address those challenges.  
In this work, we propose a framework that is tailored to leverage the capabilities of the different interfaces in heterogeneous federations.
Based on this framework, our implementation reduces the number of requests by almost 25\% leading to a tenfold decrease in query execution time over the naive approach for the example query.

\section{Federations of Linked Data Fragment Services}
\label{sec:ldf_federations}
Following existing works~\cite{DBLP:journals/ws/VerborghSHHVMHC16,DBLP:conf/semweb/HartigLP17}, we introduce a formalization of Linked Data Fragment (LDF) interfaces based on the SPARQL  expressions they are able to evaluate and the metadata they provide.
%
% LDF Original Publication
Verborgh et al. \cite{DBLP:journals/ws/VerborghSHHVMHC16} define a Linked Data Fragment (LDF) for an RDF graph $G$ as a tuple consisting of a URI, a selector function, a set of RDF triples that are the result of applying the selector function over $G$, metadata in the form of a set of RDF triples, and a set of hypermedia controls.
%The proposed notation allows for defining the LDF interface of any RDF-based data source on the Web.
%
% Linked Data Fragment Machines
Based on this work, Hartig et al. \cite{DBLP:conf/semweb/HartigLP17} propose a formal framework for comparing LDF interfaces in terms of expressiveness, complexity, and performance when evaluating SPARQL queries over different interfaces.
The concept of LDF interfaces by  Hartig et al. comprises the following:
\begin{enumerate*}[label=\roman*)]
    \item a notion of a \emph{server language} to differentiate between different capabilities of LDF interfaces, and
    \item an evaluation function in which an LDF interface provides a set of SPARQL solution mappings upon a request.  
\end{enumerate*}
%\todo{For simplicity, we also assume LDF interfaces to produce solution mappings and omit intermediate steps that may be necessary to get from a SPARQL expression to its solution mapping for a particular interface. -> maybe this is too much info for DB}  
%For example, SaGe \cite{DBLP:conf/www/MinierSM19} requires a specific client to handle the preemption of the server and to obtain the full solution mappings to a query. 
%Next, we adapt the server language definition of Hartig et al. and define it as the SPARQL expressions the LDF interface is able to evaluate. 
We adapt the server language definition from \cite{DBLP:conf/semweb/HartigLP17} to be based on the SPARQL expressions an LDF interface can evaluate. 
Therefore, we first revise the SPARQL expressions considered in the literature.  

% Interface Language 
Let the sets of RDF terms $U$, $B$, and $L$ be pairwise disjoint sets of URIs, blank nodes, and Literals, and $V$ be a set of variables disjoint from $U$, $B$, and $L$.
A triple $(s, p, o) \in (U \cup B) \times U \times (U \cup B \cup L)$ is called an RDF triple. A \emph{set} of RDF triples is an RDF graph $G$ and the universe of RDF graphs is denoted as $\mathcal{G}$. 
Following the notation by Peréz et al. \cite{DBLP:journals/tods/PerezAG09} and Schmidt et al. \cite{DBLP:conf/icdt/Schmidt0L10}, SPARQL expressions are constructed using the operators \textsc{And}, \textsc{Union}, \textsc{Optional}, \textsc{Filter}, and \textsc{Values} and can be defined recursively as follows.

\begin{definition}[SPARQL Expression]
\label{def:sparql_expression}
A SPARQL expression is an expression that is recursively defined as follows.
\begin{enumerate}
    \item A triple pattern $tp \in (U \cup V ) \times (U \cup V) \times (U \cup L \cup V)$ is a SPARQL expression ~\cite{DBLP:journals/tods/PerezAG09},
    \item if $P_1$ and $P_2$ are SPARQL expressions, then the expressions $(P_1\; \textsc{And}\; P_2 )$, $(P_1\; \textsc{Union}\; P_2 )$ and $(P_1\; \textsc{Optional}\; P_2 )$ are SPARQL expressions (\textit{conjunctive expression}, \textit{union expression}, \textit{optional expression})~\cite{DBLP:journals/tods/PerezAG09},
    \item if $P$ is a SPARQL expression and $R$ is a SPARQL filter condition, then the expression $P\ \textsc{Filter}\; R$ is a SPARQL expression (\textit{filter expression})~\cite{DBLP:journals/tods/PerezAG09}, 
    \item if $P$ is a SPARQL expression and $D$ is a SPARQL values datablock, the expression $P\ \textsc{Values}\; D$ is a SPARQL expression (\textit{values expression}), and
    \item if $P$ is SPARQL expression and $S \subset V$ is a set of variables, the expression $\textsc{Select}_S (P)$  is  an expression (\textit{select query})~\cite{DBLP:conf/icdt/Schmidt0L10}. %and $\textsc{Ask} (P)$ is an expression (\textit{ask query})~\cite{DBLP:conf/icdt/Schmidt0L10}.% (\S Def. 2). 
\end{enumerate}
\end{definition}

Furthermore, we denote the universe of SPARQL expressions as $\mathcal{P}$ and $vars(P) \subset V$ as the set of variables in the expression $P$.
We can define the interface languages of different LDF interfaces by means of the fragment of SPARQL expressions they can evaluate.

\begin{definition}[Interface Language]
Let $\mathcal{L}$ be the universe of interface languages, an interface language $L \in \mathcal{L}$ is the fragment of SPARQL expressions that an interface can evaluate. 
\end{definition}

Moreover, we denote $P \in L$ if a SPARQL expression $P$ is a SPARQL expression that is part of an interface language $L$.
Common interface languages can, thus, be defined in the following way.
\begin{itemize}
    \item $L_{\textsc{CoreSparql}}$: Any SPARQL expression defined in \cref{def:sparql_expression}.
    \item $L_{\textsc{Bgp}}$: Conjunctive expressions: $(P_1\; \textsc{And}\; P_2 )$ where $P_1$ and $P_2$ are either conjunctive expressions or triple patterns,
    \item $L_{\textsc{Tp}}$: Triple patterns.
    %\item $L_{\textsc{Tp+Filter}}$: Triple patterns and Filter expressions of the form $P\; \textsc{Filter}\; R$, where $P$ is a triple pattern.
    \item $L_{\textsc{Tp+Values}}$: Triple patterns and values expressions of the form $P\; \textsc{Values}\; D$, where $P$ is a triple pattern.
\end{itemize}
The definition of interface languages based on SPARQL expression also allows for defining containment relations between the different languages according to their expressiveness. 

\begin{definition}[Interface Language Containment]
Let $L_1$ and $L_2$ be two interface languages, we say that $L_1$ is contained in $L_2$, if all SPARQL expressions in $L_1$ are also in $L_2$: 
$$
L_1 \subseteq L_2, \text{if } \forall P \in L_1 \Rightarrow P \in L_2 .
$$
\end{definition}

For example, we can state the following containment relations for the previously introduced languages:  $L_{\textsc{Tp}} \subseteq L_{\textsc{BGP}} \subseteq L_{\textsc{CoreSparql}}$, or $L_{\textsc{Tp}} \subseteq L_{\textsc{Tp+Values}}$.
%
%The notation also allows to define more restrictive server languages, such as $L_{\textsc{Dump}}$ which support the evaluation of SPARQL expressions $P \in V \times V \times V$, or $L_{textsc{BGP}}$ which supports the evaluation of conjunctive expressions of the form $P_1\; \textsc{And}\; P_2$, where P_1 and P_2 are either triple patterns or conjunctive expressions. 
With this formalism, we can define the languages of common LDF interfaces and compare them according to their expressiveness.
Triple Pattern Fragment (TPF) servers support querying triple patterns ($L_{\textsc{Tp}}$) ~\cite{DBLP:journals/ws/VerborghSHHVMHC16}, bindings-restricted TPF server support triple patterns and \textsc{Values} expressions ($L_{\textsc{Tp+Values}}$) \cite{DBLP:journals/corr/HartigA16}, and SPARQL endpoint support any expression ($L_{\textsc{CoreSparql}}$).  

To complement the definition of LDF interfaces, we introduce the concept of \emph{interface metadata}.
For a SPARQL expression $P$, an interface may provide interface-specific metadata $M(P)$ describing the data obtained from the RDF graph.
The interface metadata may range from simple statistics such as the number of expected results to more elaborate metadata describing statistics, provenance, and licensing information.
Similar to \cite{DBLP:journals/ws/VerborghSHHVMHC16}, we assume the metadata provided for a given expression $P$ to be an RDF graph, that is $M: \mathcal{P} \to \mathcal{G}$.
%$M: \mathcal{P} \to 2^{\mathcal{T}}$.
%
Examples for common interface metadata are:
\begin{itemize}
    \item SPARQL endpoints $M_{\textsc{Ep}}$: $M(P) = \emptyset,\; \forall P \in L_{\textsc{CoreSparql}}$,
    \item Triple Pattern Fragments $M_{\textsc{Tpf}}$: $M(P)$ is an RDF graph that contains an estimate of the number of triples that match the expression $P$, $\forall P \in L_{\textsc{Tp}}$,
    \item Bindings-restricted Triple Pattern Fragments $M_{\textsc{brTpf}}$: $M(P)$ is an RDF graph that contains an estimate of the number of triples that match the expression $P$, $\forall P \in L_{\textsc{Tp+Values}}$.
\end{itemize}
Besides enabling a more fine-grained distinction between LDF interfaces, the metadata may impact the potential querying strategies employed by a client.
%
% Define the metadata for common interfaces. Is this necessary? --> YES
%
% Example: HDT-Stats --> TPF-Stats
%
Finally, combining interface language and metadata, we define a Linked Data Fragment interface as follows.

\begin{definition}[Linked Data Fragment Interface]
A Linked Data Fragment interface is a 2-tuple $f = (L_f, M_f)$, where 
\begin{itemize}
    \item $L_f \in \mathcal{L}$, the interface language,
    \item $M_f: \mathcal{P} \to \mathcal{G}$, the interface metadata for an expression $P$.
\end{itemize}
\end{definition}

%The main focus of our formalism are the \emph{capabilities} of the interfaces with respect to the expressiveness of the expressions they support and the provided metadata. 
Conceptually, we distinguish LDF \emph{interfaces} which define the interface language and metadata, and LDF \emph{services}, which are Web servers that implement a specific interface.

\begin{definition}[Linked Data Fragment Service]
A Linked Data Fragment service $c \in U$ is a Web service that supports the evaluation of SPARQL expressions and provides metadata according to the LDF interface $int(c) = (L_c, M_c)$ that it implements.
\end{definition}

%For simplicity, we denote $P \in int(c)$ if an expression $P$ is a SPARQL expression in the interface language of the interfaces $L_c$.
We reuse the function $ep: U \to \mathcal{G}$~\cite{DBLP:conf/esws/ArandaAC11} that maps an LDF service to the RDF graph $ep(c)$ available at the service.
The evaluation of a SPARQL expression $P$ over an LDF service $c$ is then given as
\begin{equation}
\label{eq:1}
\llbracket P \rrbracket_{c} := %\llbracket P \rrbracket_{ep(c)} .
\begin{cases}
\llbracket P \rrbracket_{ep(c)}, & \text{if } P \in L_c. \\
\emptyset, &  \text{otherwise.}
\end{cases}
\end{equation}

Note the difference in the subscript $c$ and $ep(c)$ to distinguish between solution mappings produced by an LDF service $\llbracket \cdot \rrbracket_c$ regarding its interface language and the solution mappings for evaluating any expression over the graph available at the LDF service  $\llbracket \cdot \rrbracket_{ep(c)}$.
%The former defines the evaluation according to the \textit{capability} of the LDF service's interface while the latter defines what the solution mappings are for any SPARQL expression over the graph of the service.
% Add an example here:
%
%
Combining these previous definitions, we define the concept of federations of LDF services as follows.

\begin{definition}[Federation of Linked Data Fragment Services]
A Federation of Linked Data Fragment services is a 3-tuple $F = ( C, int, ep)$, where 
\begin{itemize}
    \item $C = \{c_1, \dots, c_n\} \subset U$, a set of URIs for LDF services,
    \item $int$, a function that maps an LDF service to its interface,
    \item $ep$, a function that maps each LDF service to the graph available at that service.
\end{itemize}
\end{definition}

Federations in which all LDF services implement the same LDF interfaces are called \textit{homogeneous}, and \textit{heterogeneous} otherwise.
%Following this notation, we can define a ``traditional'' federation of SPARQL endpoints as a homogeneous federation $F = (C, int, ep)$ with $L_c = L_{\textsc{CoreSparql}}, \forall c \in C$
%
For practical reasons, in the remainder of this work, we just consider graphs in the federation without blank nodes and focus on federations in which all members are at least able to evaluate triple patterns of any form: $L_{\textsc{TP}} \subseteq L_c,\, \forall c \in C$.
\footnote{This means that we do not include data dumps, even though they are also considered Linked Data Fragments in other works \cite{DBLP:conf/semweb/HartigLP17,DBLP:journals/ws/VerborghSHHVMHC16}.}
\begin{example}
We can define the federation from our motivating example as $F_{ex} = (\{c_1, c_2\}, int, ep)$ with  $c_1=$ \texttt{wikidata:sparql}\textsuperscript{\ref{fn:wd_sparql}}, $c_2=$ \texttt{dbpedia:tpf}\textsuperscript{\ref{fn:dbp_tpf}}, $int(c_1) = (L_{\textsc{CoreSparql}}, M_{\textsc{Ep}})$, $int(c_2) = (L_{\textsc{Tp}},$ $M_{\textsc{Tpf}})$, $ep(c_1) = G_{Wikidata}$ and $ep(c_2) = G_{DBpedia}$. 
\end{example}
Following the notation by Acosta et al.~\cite{DBLP:reference/bdt/AcostaHS19}, we denote the evaluation of a SPARQL expression over a federation of LDF interfaces $F$ as $\llbracket \cdot \rrbracket_F$ and define the semantics in the following way. 

% Semantics
\begin{definition}[Set Semantics of SPARQL Query Processing over LDF Service Federations]
\label{def:set_semantics}
Given a SPARQL expression $P$ and a federation $F = (C, int, ep)$, the result set of evaluating $P$ over $F$ is given as
$$
\llbracket P \rrbracket_F := \llbracket P \rrbracket_G \text{, with } G = \bigcup \limits_{\forall c \in C} ep(c)
$$
\end{definition}

% NOTE: Add the formalism of the federation in the motivating example here?

%In other words, the result of evaluating a SPARQL expression $P$ over a federation of LDF services is the same as evaluating $P$ over the union of the graphs available at the services of the federation. 
%This definition differentiates our focus on federated query processing to \emph{alternative-aware} query processing, in which the goal is to exploit the strengths of different interfaces providing access to same graph \cite{DBLP:conf/semweb/MontoyaAH18a}. 

\section{Federated Query Processing over Heterogeneous Federations}
\label{sec:approach}
In the presence of heterogeneous LDF service federations, novel challenges arise that cannot be addressed by existing approaches. % for homogeneous federations of SPARQL endpoints.
Therefore, we propose a framework for heterogeneous federations to address the central components of federated query processing: 
\begin{enumerate*}
    \item[(\S\ref{sec:query_decomposition})] query decomposition,
    \item[(\S\ref{sec:query_planner})]  query planning, and
    \item[(\S\ref{sec:polymorphic_join_operator})] physical operators.
\end{enumerate*}
Furthermore, for each component of the framework, we propose an approach aiming to obtain efficient query plans.%, which minimize execution time and reduce the number of requests to the services in the federation by exploiting their capabilities.

\subsection{Query Decomposition}
\label{sec:query_decomposition}

The goal of query decomposition is grouping the query into subexpressions such that the evaluation of the subexpressions over the members of the federation minimizes execution time while ensuring that all expected answers are produced. 
Existing decomposition approaches assume that all federation members are able to evaluate any SPARQL expression.
Since this assumption is not valid in heterogeneous federations, we propose \textit{interface-compliant} query decompositions and their evaluation over such federations.

Given a given SPARQL query $P$ and a federation $F = (C, int, ep)$, query decomposition aims to group a query into subexpressions such that they can be answered by the relevant sources in the federation. 
% Source Selection
The first step to achieve this goal is source selection.
That is, select the relevant sources $r(tp_i)$ for all triple patterns $tp_i$ in $P$, with $r(tp_i) = \{ c \in C \mid \llbracket tp_i \rrbracket_{ep(c)} \neq \emptyset \}$.
Because we require all LDF services to at least evaluate triple patterns of any form ($L_{\textsc{TP}} \subseteq L_c,\, \forall c \in C$), in principle, the relevant sources can be selected by evaluating each triple pattern at each service.
Typically, the capabilities of the services allow more efficient source selection strategy implementations, such as \textsc{Ask} queries for SPARQL endpoints, leveraging the metadata  $($u$, \texttt{void:triples}, cnt) \in M(tp_i)$ for TPF and brTPF servers, or using pre-computed data catalogues.
% 

% Query Decomposition 
Once the relevant sources are identified, the query engine decomposes the query into subexpressions to be evaluated at the services in the federations.  
If there exists a triple pattern $tp$ in a basic graph pattern (BGP) $P$ with no relevant source $r(tp) = \emptyset$, the evaluation of $P$ over the federation is the empty set.
In the following, we focus on query decompositions for BGPs where all triple patterns have at least one relevant source.
For simplicity, we extend notation and consider a BGP $P = (tp_1\, \textsc{And}\, \dots\, \textsc{And}\, tp_n)$ also as a set of $n$ triple patterns: $P =\{tp_1, \dots, tp_n\}$.
%\footnote{For simplicity, we abuse notation and consider a BGP $P = (tp_1\, \textsc{And}\, \dots\, \textsc{And}\, tp_n)$ also as a set of $n$ triple patterns: $P =\{tp_1, \dots, tp_n\}$.} 
%
\begin{definition}[Query Decomposition]
Given a BGP $P$ and an LDF service federation $F= (C, int, ep)$, a query decomposition $D(P, F) = \{ d_1, \dots, d_m \}$ is a set of tuples  $d_i = ( SE_i, S_i )$ where
\begin{itemize}
    \item $SE_i$ is a subexpression of $P$, and
    \item $S_i \subseteq C $ a non-empty the subset of services over which $SE_i$ is evaluated, such that $\bigcup \limits_{i=1, \dots , n} SE_i = P$.
\end{itemize}
 % Reason for these condition: If we do not have these restrictions, the approach for assessing the cost does not work properly
\end{definition}
% Note: Here we are abusing the notation a bit, and use subexpressions as sets of triple patterns. Maybe we need to introduce this

Because the query decomposition as such does not consider the interface language of the services, it is possible that for a valid query decomposition $D$:
$
\exists d_i \in D: \exists c \in S_i: SE_i \not \in L_c.
$
Considering the query from our motivating example, a valid decomposition would be $D(P, F) = \{((tp_1\, \textsc{And}\, tp_2), \{c_1\}), ((tp_3\, \textsc{And}\, tp_4\, \textsc{And}\, tp_5), \{c_2\})\}$, even though $c_2$ is a TPF server that can only evaluate triple patterns. 
%For instance, consider a decomposition $D = \{ (SE_1, S_1) \}$ with the subexpression $SE_1 = (tp_1\; \textsc{And}\; tp_2)$ and $S_1 = \{ c \}$, where the only relevant source is the TPF service $c$ with $int(c) = (L_{\textsc{Tp}}, M_{\textsc{TPF}})$.
%Since $SE_1 \not \in L_{\textsc{Tp}}$, we cannot evaluate $SE_1$ at service $c$, that is $\llbracket SE_1 \rrbracket_{c} = \emptyset$.
%To overcome this issue but without requiring a more complex definition of the query decomposition
Therefore, we introduce an evaluation function $\theta$ for the interface-compliant evaluation of SPARQL expressions. % in an interface languages compliant manner.

\begin{definition}[Interface-compliant Evaluation of an Expression]
\label{def:interface_compliant_evaluation}
Given a BGP $P$ and an LDF service $c \in U$, the interface-compliant evaluation of $P$ over $c$ is given as follows.
\begin{numcases}{\theta_c(P) :=}
\llbracket P \rrbracket_{c} & if $P \in L_c$. \\
\llbracket P_ 1 \rrbracket_{c} \bowtie \dots \bowtie \llbracket P_l \rrbracket_{c} & otherwise. \label{case:not_in_language} 
\end{numcases}

For some $P_1, \dots, P_l$ with $P = (P_1\; \textsc{And}\; \dots\; \textsc{And}\; P_l )$ in Equation (\ref{case:not_in_language}), such that the following conditions hold:
\begin{itemize}
    \item Equivalence: $\theta_c(P) \equiv  \llbracket P \rrbracket_{ep(c)}$ %\llbracket tp_1  \rrbracket_{ep(c)} \bowtie \dots \bowtie  \llbracket tp_n  \rrbracket_{ep(c)}$
    \item Compliance: $P_i \in L_c,\; \forall P_i \text{ in } \{P_1  \dots\  P_l \}$%(P_1\; \textsc{And}\; \dots\; \textsc{And}\; P_l )$
\end{itemize}
\end{definition}

The intuition of the interface-compliant evaluation is as follows.
If the BGP $P$ is in the language of the service, then $P$ can be evaluated directly at the service.
Otherwise, the original expression $P$ is split into subexpressions, such that each subexpression is in the language of the service. 
As a result, joining the solutions of evaluating the individual subexpressions at the service yields the same solutions as evaluating $P$ over the graph of the service $ep(c)$.
We denote the number of subexpressions in the interface-compliant evaluation as $|\theta_c(P)|$.
A compliant evaluation of $SE= (tp_3\, \textsc{And}\, tp_4\, \textsc{And}\, tp_5 )$ from our previous example at the DBpedia TPF server $c_2$ would be $\theta_c(SE) = \llbracket tp_3 \rrbracket_{c_2} \bowtie \llbracket tp_4 \rrbracket_{c_2} \bowtie \llbracket tp_5 \rrbracket_{c_2}$.
With this notion, we define the interface-compliant evaluation of a query decomposition. % in the following way. 

\begin{definition}[Interface-compliant Evaluation of a Query Decomposition]
\label{def:compliant_decomp_eval}
Given a query decomposition $D(P, F)$ for the BGP $P$ and federation $F$, the evaluation of $P$ following the query decomposition $\theta_{D(P, F)}(P) $ is given as the conjunction ($\Bowtie$)  of the subexpressions $SE_i$ evaluated at all ($\cup$) services in $S_i$:
$$
\theta_{D(P, F)}(P): =\;  \Bowtie_{(SE_i, S_i) \in D(P, F)}( \cup_{c_j \in S_i}\; \theta_{c_j}(SE_i) )
$$
\end{definition}

% Add an example here, ideally from the motivating example
After defining query decompositions and the evaluation of such decompositions that is compliant with the interfaces of the LDF services in the federations, the problem of finding a \textit{suitable} query decomposition for a given query arises.  
The common goal of query decomposition approaches is finding a decomposition that yields complete answers according to the assumed semantics, while the cost of executing the decomposition by the query engine is minimized  \cite{DBLP:conf/dexa/EndrisGLMVA17,DBLP:journals/tlsdkcs/VidalCAMP16}. %i.e, $\theta_{D(P, F)}(P) = \llbracket P \rrbracket_{F}$
However, these approaches do not explicitly measure the expected answer completeness of query decompositions.
For instance, in \cite{DBLP:journals/tlsdkcs/VidalCAMP16} the answer completeness is encoded implicitly in the query decomposition cost by considering the number of non-selected endpoints: if fewer relevant endpoints are contacted according to a decomposition, its cost is higher and vice versa. 
Extending existing approaches, we propose the concept of query decomposition \emph{density} as a measure to estimate and compare the expected answer completeness of different decompositions.
In contrast to \cite{DBLP:journals/tlsdkcs/VidalCAMP16}, our density measure not only considers the non-selected endpoints but also how triple patterns are grouped into subexpressions that are evaluated jointly at the services.
% Notes: 
%Further, they ensure to only submit relevant subqueries to the endpoints by using ASK queries when deciding which subquery should be evaluated at which source. This, however, is only possible if all sources support that 
% Also, in our case the same 

% Maybe we can leave this part out
%This is due two main reasons: 
%\begin{enumerate*}[label=\roman*)]
%    \item the graphs of available at the members of the federation are disjoint, and
%    \item the decomposition approaches leverage more fine-grained statistics than just the information on relevant sources in order to prune irrelevant sources whose answers to not contribute to the overall answers of the query. 
%\end{enumerate*}

\smallbreak
\noindent 
\textit{Query Decomposition Density. }
The query decomposition density is a proxy for the expected answer completeness of a decomposition.
%The core idea is to determine the similarity of a given decomposition to a decomposition that guarantees answer completeness.
We define density as a relative measure with respect to a decomposition that guarantees answer completeness, i.e., the \emph{atomic} decomposition.
The atomic decomposition evaluates every single triple pattern in a subexpression at all relevant sources and thus, guarantees answer completeness.
%The core idea is to determine the similarity of a given decomposition to the naive yet \emph{atomic} decomposition that evaluates every single triple pattern in a subexpression at all relevant sources and, thus, its evaluation guarantees complete answers.

\begin{definition}[Atomic Decomposition]
\label{def:base}
Given a federation $F = (\{ c_1, \dots, c_k \}, int, ep)$ and a BGP $P = (tp_1\; \textsc{And}\; \dots \textsc{And}\; tp_n)$, the atomic decomposition is given as
$$
D^*(P, F) = \{ (tp_1 , r(tp_1)), \dots , (tp_n, r(tp_n) \}.
$$
\end{definition}

\vspace{1em}
\begin{lemma}%[Answer Completeness of $D^*(P, F)$]
The evaluation of $P$ following  $D^*(P,F)$ yields complete answers, that is:
\begin{equation}
\label{eq:completenes}
\theta_{D^*(P, F)}(P) = \llbracket P \rrbracket_{F}    
\end{equation}
\end{lemma}

% ========================================
% Start old proof
% ========================================
\begin{proof}
We provide a direct proof by assuming the left-hand side in \cref{eq:completenes}.  
Since we require all services to be able to evaluate triple patterns and $D^*$ is composed of triple patterns only, the evaluation of $D^*(P, F)$ is given by \cref{def:interface_compliant_evaluation} and \cref{def:compliant_decomp_eval} as 
\begin{equation}
\label{eq:2}
\theta_{D^*(P, F)}(P): =\;  \Bowtie_{(SE_i, S_i) \in D^*(P, F)}( \cup_{c_j \in S_i}\; \llbracket SE_i \rrbracket_{c_j} )
\end{equation}
with $SE_i = tp_i$.
By Def.~\ref{def:base}, $S_i$ corresponds to the relevant sources of $tp_i$, which is given by $r(tp_i) = \{ r_{i1}, \dots, r_{im} \}$. 
Next, we expand the \cref{eq:2} with $S_i$ in the following way.
%$r(tp_i) = \{ r_{(i,1)}, \dots, r_{(i,m)} \}$

\begin{equation}
\label{eq:3}
%\footnotesize
(\llbracket tp_1 \rrbracket_{r_{11}} \cup \dots \cup \llbracket tp_1 \rrbracket_{r_{1l}} ) \bowtie \dots \bowtie (\llbracket tp_n \rrbracket_{r_{n1}} \cup \dots \cup \llbracket tp_n \rrbracket_{r_{no}} )
\end{equation}

Next, we show that we can evaluate all triples patterns at all sources (relevant and non-relevant).
By definition, we have that the evaluation of a triple pattern over a non-relevant source is the empty set:  $\llbracket tp_i \rrbracket_{c} = \emptyset,\; \forall c \not \in r(tp_i)$. Further, since $( \llbracket tp_i \rrbracket_{r_{ij}} \cup \emptyset) = \llbracket tp_i \rrbracket_{r_{ij}}$, we can expand \cref{eq:3} to 

\begin{equation}
\label{eq:4}
(\llbracket tp_1 \rrbracket_{c_1} \cup \dots \cup \llbracket tp_1 \rrbracket_{c_k} ) \bowtie \dots \bowtie (\llbracket tp_n \rrbracket_{c_1} \cup \dots \cup \llbracket tp_n \rrbracket_{c_k} )
\end{equation}

According to \cref{eq:1} and the fact that triple patterns are in the interface language of all services, we have $\llbracket tp_i \rrbracket_{c_j} = \llbracket tp_i \rrbracket_{ep(c_j)}$ and can rewrite \cref{eq:4} as 

\begin{equation}
\footnotesize
\label{eq:5}
(\llbracket tp_1 \rrbracket_{ep(c_1)} \cup \dots \cup \llbracket tp_1 \rrbracket_{ep(c_k)} ) \bowtie \dots \bowtie (\llbracket tp_n \rrbracket_{ep(c_1)} \cup \dots \cup \llbracket tp_n \rrbracket_{ep(c_k)} )
\end{equation}
Because we assume set semantics, the following equality holds\footnote{
We prove this equality by contradiction. 
Consider $G = \bigcup_{c \in C} ep(c)$.  
Assume that there exists a solution mapping $\mu$ s.t. $\mu \in (\llbracket tp_i \rrbracket_{ep(c_1)} \cup \dots \cup \llbracket tp_i \rrbracket_{ep(c_k)} )$ 
and $\mu \notin \llbracket tp_i \rrbracket_G$. 
This means that the evaluation of a subexpression over some source, e.g. $\llbracket tp_i \rrbracket_{ep(c_j)}$, is producing additional answers w.r.t. the evaluation over the union of all RDF graphs. 
This could only happen if $ep(c_j) \nsubseteq G$, however,  this contradicts the definition of $G$. % at least one triple in one of the individual graphs $ep(c_i)$ that are not part of the union of all individual graphs $\bigcup_{c \in C} ep(c)$.
Now assume that $\mu \in \llbracket tp_i \rrbracket_G$ but $\mu \notin (\llbracket tp_i \rrbracket_{ep(c_1)} \cup \dots \cup \llbracket tp_i \rrbracket_{ep(c_k)} )$.  
Without loss of generality, assume that $\mu$ was produced from matching an RDF triple $t \in  G$  s.t. $t \notin ep(c)$ for all services $c \in C$ in the federation. This is again a contradiction with the definition of $G$.}

\begin{equation}
\label{eq:6}
(\llbracket tp_i \rrbracket_{ep(c_1)} \cup \dots \cup \llbracket tp_i \rrbracket_{ep(c_k)} ) = \llbracket tp_i \rrbracket_{\bigcup \limits_{c \in C} ep(c)}
\end{equation}
With \cref{eq:6} we can reformulate \cref{eq:5} as 
\begin{equation}
\label{eq:7}
\llbracket tp_1 \rrbracket_{\bigcup \limits_{c \in C} ep(c)} \bowtie \dots \bowtie  \llbracket tp_n \rrbracket_{\bigcup \limits_{c \in C} ep(c)}
\end{equation}
and according to \cref{def:set_semantics} and Definition 4 in \cite{DBLP:conf/icdt/Schmidt0L10}, we have the following equality:

\begin{equation*}
%\label{eq:8}
\begin{aligned}
& \llbracket tp_1 \rrbracket_{\bigcup \limits_{c \in C} ep(c)} \bowtie \dots \bowtie  \llbracket tp_n \rrbracket_{\bigcup \limits_{c \in C} ep(c)} \\
=\; & \llbracket tp_1 \rrbracket_{F} \bowtie \dots \bowtie  \llbracket tp_n \rrbracket_{F} \\
=\; & \llbracket tp_1\; \textsc{And}\; \dots \;\textsc{And}\; tp_n \rrbracket_{F}  \\
=\; & \llbracket P \rrbracket_{F}
\end{aligned}
\end{equation*}
\end{proof}
% ========================================
% End old proof
% ========================================

%\input{figures/decomposition_completeness}
%Next, we need to determine the similarity between the atomic decomposition and any other decomposition. 
%To this aim, we first define the concept exclusive groups similar to Schwarte et al. \cite{DBLP:conf/semweb/SchwarteHHSS11}.

% Exclusive group
Another type of structure that preserves completeness are exclusive groups \cite{DBLP:conf/semweb/SchwarteHHSS11}, which are subexpressions of a query that can only be answered by a single source. 
They are defined as follows.
\begin{definition}[Exclusive Group]
Given a federation $F= (C, int, ep )$, a BGP $X$ is called an exclusive group, if for all triple patterns $tp_i \in X$ there exists only one relevant source $c_{X} \in C$:
$$
X = \{ tp_i \mid r(tp_i) = \{ c_{X} \}  \}
$$
\end{definition}

%In order to assess the completeness of a query decomposition, we represent it by a \textit{decomposition graph} and compare its structure to the decomposition graph of the atomic decomposition $D^*(P,F)$.
We represent query decompositions by \textit{decomposition graphs} and compute the relative density with respect to the decomposition graph of the atomic decomposition $D^*(P,F)$ as a measure of completeness.
More edges in the graph of a given decomposition yield a higher density and, thus, expected answer completeness. 

\begin{definition}[Query Decomposition Graph]
\label{def:decomposition_graph}
Let $D(P,F)$ be a query decomposition for the BGP $P$ and federation $F=(C, int, ep)$.
The decomposition graph $G_{D(P, F)} = (V, E)$ of $D(P,F)$ is: % defined as: %with $V$ the set of vertices and $E \in V \times V$ the set of edges in $G$.
\begin{itemize}[leftmargin=0.8\parindent]
    \item[] The set of vertices $V = \{ tp_i \in P\} \cup \{ r(tp_i) \mid \forall tp_i \in P \}$,
    \item[] The set of edges $E \subseteq V \times V$ are given by the following rules.
\end{itemize}
\begin{enumerate}[label=Rule \Roman*,leftmargin=4\parindent]
    \item \label{item:rule1} Add an edge between a triple pattern $tp_i \in P$ and a relevant source $r_{ij} \in r(tp_i$), if $tp_i$ is part of a subexpression $SE$ that is evaluated at $r_{ij}$: $\exists (SE, S) \in D(P,F)$ with $tp_i \in SE \land r_{ij} \in S$.
    \item \label{item:rule2} Add an edge between two triple patterns $tp_i$ and $tp_j$, if they do not co-occur in a subexpression $SE$ in $D$: $(tp_i, tp_j) \in E$, if $\not \exists \, SE \in D(P, F): tp_i \in SE \land tp_j \in SE$.
    \item \label{item:rule3} Add an edge between two triple patterns $tp_i$ and $tp_j$, if they are part of the same exclusive group $X$: $(tp_i, tp_j) \in E$, if $tp_i \in X \land tp_j \in X$.
    \item \label{item:rule4} Add an edge between all triple patterns, if the decomposition is composed of just a single subexpression to be evaluated at one source : $D(P, F) = \{ (SE, S ) \} \land |S| = 1 $. 
\end{enumerate}
\end{definition}
% Note: Check Rule IV: Could be obsolete because of rule III: No, because there could be another relevant source that is not in S. But if just a single source is contacted, the completeness is not impacted depending on the subexpressions 

The rules for adding edges to the graph are designed in such a way that the maximum number of edges is present for the decomposition graph of the atomic decomposition $G_{D^*(P, F)} = (V^*, E^*)$.
This is because each triple pattern is connected to each relevant source (\ref{item:rule1}) and there is an edge between each pair of triple patterns (\ref{item:rule2}).
If a decomposition contacts fewer sources, the decomposition graph will have fewer edges according to \ref{item:rule1}.
Further, if more triple patterns are grouped together into subexpressions in a query decomposition, its graph will also have fewer edges according to \ref{item:rule2}. 
The rationale of this rule is that grouping triple patterns could potentially miss solution mappings that are only produced by joining data from two different sources.
The remaining rules are introduced to handle the following exceptions.
\ref{item:rule3} handles exclusive groups: triple patterns of exclusive groups can be grouped into a single subexpression without negatively impacting the answers completeness.
Finally, \ref{item:rule4} handles the following cases: if the decomposition only has a single subexpression that is evaluated at a single source, it does not have an impact on the completeness how these triple patterns are grouped into subexpressions of the decomposition. 
% Short explanation of the difference:
In contrast to \ref{item:rule3}, in \ref{item:rule4} even though $SE$ is just evaluated at a single source, $SE$ does not need to be an exclusive group and can have other relevant sources that are not in $S$.
By these rules, we can measure the density of a decomposition graph relative to the maximum number of possible edges as given by the atomic decomposition graph.

\begin{definition}[Density of a Query Decomposition]
\label{def:decomposition_completeness}
Given a query decomposition $D(P, F)$ and the corresponding graph $G_{D(P,F)} = (V, E)$, the density $density(D(P, F))$ is computed as:
$$
density(D(P, F)) = \frac{|E|}{|E^*|} \in [0,1] .
$$
\end{definition}
%
% Number of edges in atomic decomposition:
% E^* := 1/2 * |P| * (|P| - 1 ) + \sum \limits_(tp_i \in P) r(tp_i)
%
\begin{theorem}%[Answer Completeness of $density(D(P, F))= 1.0$]
The evaluation of a query decomposition $D(P, F)$ over a federation $F$ yields complete answers, if $density(D(P, F)) = 1$:
\begin{equation}
\label{eq:completenes_implication}
density(D(P, F)) = 1 \implies \theta_{D(P,F)}(P) = \llbracket P \rrbracket_F
\end{equation}
\end{theorem}

\begin{proof}
We prove the implication in \cref{eq:completenes_implication} by contradiction.
We assume $density(D(P, F))=1$ and $\theta_{D(P,F)}(P) \neq \llbracket P \rrbracket_F$.
According to \cref{def:decomposition_completeness}, $density(D(P, F)) = 1$ holds only if the decomposition graph of $D(P,F)$ has the same number of edges as the decomposition graph of the atomic decomposition: $|E| = |E^*|$.
Following \cref{def:base} and \cref{def:decomposition_graph}, in $E^*$ there is an edge between each triple pattern and its relevant sources (\ref{item:rule1}) and an edge between every pair of triple patterns (\ref{item:rule2}).
The maximum number of edges is
$$
|E^*| = \underbrace{\sum_{tp_i \in P} r(tp_i)}_{\text{\ref{item:rule1}}}\; + \; \underbrace{\vphantom{ \sum_{tp_i \in P} r(tp_i))} 0.5 \cdot |P|  \cdot  (|P| - 1)}_{\text{\ref{item:rule2}}} . 
$$

Since we prove completeness, we focus on the case when a decomposition yields fewer answers: $\theta_{D(P,F)}(P) \subset \llbracket P \rrbracket_F$.
This can occur in two cases:
\begin{enumerate}[label=\textsc{Case} \arabic*:, leftmargin=*]%,wide]
    \item A part of the query is not evaluated at a relevant source. 
    Without loss of generality, consider that a triple pattern $tp_i \in P$ is not evaluated at a relevant source $c_j$ and $\llbracket tp_i \rrbracket_{c_j}$ contributes to the answers of $P$. 
    In this case, the decomposition graph $G_{D(P,F)}$ does not have an edge $(tp_i,c_j)$ according to \ref{item:rule1} and,  therefore, $|E| < |E^*|$. This contradicts the assumption that $density(D(P,F)) = 1$.
    
    \item Triple patterns with several relevant sources are grouped into subexpressions.  
    Consider the solution mapping $\mu \in \llbracket P \rrbracket_F$, with $\mu = \{ \mu_1 \cup \mu_2 \mid \mu_1 \in \llbracket tp_1 \rrbracket_{c_1} \land \mu_2 \in \llbracket tp_2 \rrbracket_{c_2}, \mu_1 \sim \mu_2 \}$, and without loss of generality assume that $r(tp_1) = r(tp_2) = \{c_1, c_2\}$. 
    Such a solution mapping $\mu$ does not exist in $\theta_{D(P,F)}(P)$ in the case that the two triple patterns are evaluated jointly at the source $c_1$ and $c_2$, that is 
    $$
    ( (tp_1 \; \textsc{And}\; tp_2), \{c_1, c_2\}) \in D(P,F)
    $$
    In this case, the edge $(tp_1, tp_2)$ does not exist in $E$ according to \ref{item:rule2} but the edge exists in $E^*$ because
    $$
    ( tp_1, \{c_1, c_2\} ), ( tp_2, \{c_1, c_2\} ) \in D^*(P,F)
    $$
    Therefore, we have $|E| < |E^*|$ which contradicts the assumption that $density(D(P,F)) = 1$.
\end{enumerate}
\end{proof}
% Why Rule III and Rule IV do not apply:
% - Rule III focuses on exclusive groups. If a part of the exclusive group is not evaluated at the relevant source, Case 1 applies. Case 2 does not apply for exclusive groups, because there could not be two different source c_1 and c_2
% - Rule IV: SE is not an exclusive group otherwise Rule III applies. Therefore, we have S \subset of \cup \limtits_{tp_i \in SE} r(tp_i) (S is a subset of all relevant sources) and therefore, Case 1 applies.

\captionsetup*[subfigure]{position=bottom,textfont=normalfont,labelfont=normalfont}
\begin{figure*}[t]
    \centering
    \subfloat[Graph $G_{D^*(P,F)}$% with $comp(D^*(P,F)) = 1$
    ]{\includegraphics[width=0.23\textwidth]{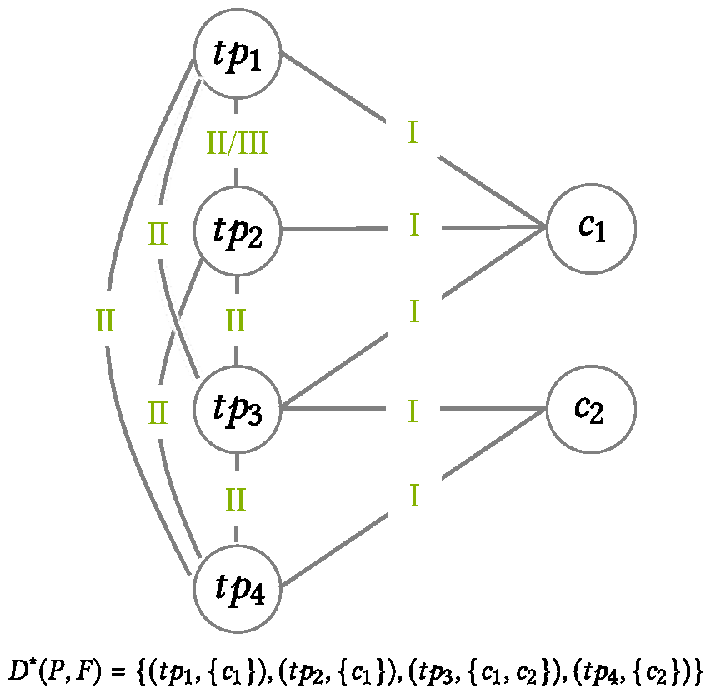}\label{fig:d_star}}
    %\hspace{4mm}
    \hfill
    \subfloat[Graph $G_{D_1(P,F)}$% with $comp(D_1(P,F)) = 1$
    ]{\includegraphics[width=0.23\textwidth]{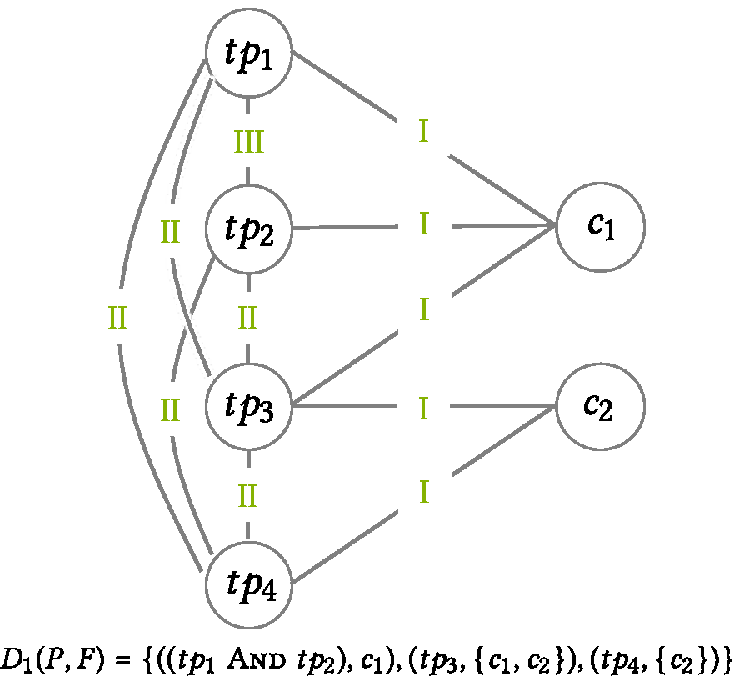}\label{fig:d_1}}
    %\hspace{4mm}
    \hfill
    \subfloat[Graph $G_{D_2(P,F)}$% with $comp(D_2(P,F)) = \frac{9}{11}$
    ]{\includegraphics[width=0.23\textwidth]{figures/d_1.eps}\label{fig:d_2}}
    %\hspace{4mm}    
    \hfill
    \subfloat[Graph $G_{D_3(P,F)}$% with $comp(D_3(P,F)) = \frac{8}{11}$
    ]{\includegraphics[width=0.23\textwidth]{figures/d_1.eps}\label{fig:d_3}}
    \vspace{-0.5\baselineskip}
    \caption{Query decomposition graphs for decompositions from \Cref{ex:decomp_graphs}. The rules for adding edges are indicated in green.}
    \label{fig:motivating_example}
\end{figure*}

We can prove that a decomposition density of $1$ implies answer completeness, however, the inverse (i.e., $density(D(P, F)) = 1 \impliedby \theta_{D(P,F)}(P) = \llbracket P \rrbracket_F $ ) cannot be guaranteed. 
For example, there might be a triple pattern $tp_1$ with two relevant sources $c_1$ and $c_2$ with just source $c_1$ contributing to the final answers. 
A decomposition $D(P,F)$ where $tp_1$ is not evaluated at $c_2$ might still yield complete answers but $density(D(P,F)) < 1$ according to \ref{item:rule1}.
Therefore, the decomposition density is a measure for the \emph{expected} completeness based on the assumptions that answer completeness is negatively affected while 
\begin{enumerate*}[label=\roman*)]
    \item contacting fewer relevant sources, and
    \item grouping triple patterns that can be evaluated at several sources into subexpressions.
\end{enumerate*}
Estimating the true completeness more accurately would require additional information on the data provided by the LDF services than just the relevant sources.
Such additional information could be used to improve the effectiveness of our measure, for example by weighting the edges in the decomposition graph according to their importance. 
However, such an extension is out of the scope of this work. 

\begin{example}
\label{ex:decomp_graphs}
Let us consider the BGP $P = (tp_1\; \textsc{And}\; tp_2\; \textsc{And}\;$ $tp_3 \; \textsc{And}\; tp_4 )$ from the SPARQL query of the motivating example in \Cref{sec:motivating_example} and the federation $F_{ex}=( \{c_1, c_2 \}, int, ep)$.
The relevant sources are $r(tp_1) = \{ c_1 \}$, $r(tp_2) = \{ c_1 \}$, $r(tp_3) = \{ c_1, c_2 \}$, $r(tp_4) = \{ c_2 \}$. 
The atomic query decomposition is $D^*(P,F) = \{ (tp_1, \{c_1\}),$ $(tp_2, \{c_1\}), (tp_3, \{c_1, c_2\}),$ $(tp_4, \{ c_2\}) \}$ and the corresponding graph is shown in \cref{fig:d_star}. In $P$, the triple patterns $tp_1$ and $tp_2$ form an exclusive group, as they are both only answerable by service $c_1$. 
Therefore, we can combine them in a single subexpression without reducing the expected completeness in $D_1(P,F) = \{ ((tp_1\; \textsc{And}\; tp_2),$ $\{c_1\}), $ $(tp_3, \{c_1, c_2\}), $ $(tp_4, \{c_1\})  \}$. 
The corresponding graph shown in \cref{fig:d_1} is identical to $G_{D^*(P,G)}$ and thus its expected completeness is $density(D_1(P,F)) = \frac{11}{11} = 1$. 
Alternatively, we can choose to evaluate $tp_3$ only at service $c_2$ with $D_2(P,F) = \{ ((tp_1\; \textsc{And}\; tp_2), \{c_1\})$ $((tp_3\; \textsc{And}\; tp_4), \{c_2\}) \}$ (\cref{fig:d_2}) or evaluate $tp_3$ at service $c_1$ with $D_3(P,F) = \{ ((tp_1\; \textsc{And}\; tp_2\; \textsc{And}\; tp_3),$ $ \{c_1\}) , (tp_4, \{c_2\}) \}$ (\cref{fig:d_3}). 
Since both corresponding decomposition graphs have fewer edges than the graph of $G_{D^*(P, F)}$, we expect fewer answers because: $density(D^*(P, F)) > density(D_2(P, F)) > density(D_3(P, F)) $.
\end{example}

\smallbreak
\noindent 
\textit{Query Decomposition Cost. } 
The example illustrates how query decompositions can have different levels of expected completeness.
Ideally, one would always choose the atomic query decomposition to guarantee complete answers. 
However, there are also costs associated with the evaluation of a decomposition that are induced by the amount of transferred data for intermediate results during query execution as well as the number of services that need to be contacted.
In federations of SPARQL endpoints, both goals are achieved by
\begin{enumerate*}[label=\roman*)]
    \item decomposing the query into as few subexpressions as possible, and
    \item reducing the number of endpoints contacted by selecting just those sources which are likely to contribute to the final answer of the query. 
\end{enumerate*} 
In contrast, when facing heterogeneous federations of LDF services, the languages of the LDF services need to be considered as well.
The reason is that the interface-compliant evaluation might yield additional costs in cases when subexpressions cannot be evaluated by a service as a whole. 
There might be several interface-compliant evaluations for an expression because the original expression could be split in different ways into subexpressions that can be evaluated by the service.
We denote an interface-compliant evaluation of an expression $P$ with the minimal number of subexpressions as  $\theta^*_c(P)$, which is the evaluation of $P$ that requires separating the expression into the fewest subexpression to be interface-compliant.
Note that $|\theta^*_c(P)| = 1$, if $P \in L_c$.
% Note:
% Complexity of finding the minimizing evaluation depends on the interface language. 
% If for example, on triple patterns can be evaluated at the source, the minimal number of sub-expression in the interface-compliant evaluation is given by the number of triple patterns in a subquery
% If we have more capable sources that allow for evaluating BGPs under a certain condition (e.g., Star Pattern Fragment), 
% we would need to check all possible partitions of the BGP into sub BGP, whose number is given by the bell number (which is exponential in n)

%In order to capture both aspects in the cost of query decomposition, we propose a cost function that considers the number service contacted, the number of subexpressions as well as the interface-compliant evaluation of the subexpressions.
%
% Cost

We propose a lower bound for query decomposition cost that considers the number of services contacted and the number of subexpressions in an interface-compliant evaluation of the decomposition.
In particular, this lower bound combines: 
\begin{enumerate*}
    \item The number of sources $|S|$ to be contacted per subexpression.
    \item The number of additional subexpressions ($|\theta^*_{c}(SE)| - 1$) required for an interface-compliant evaluation for each subexpression and all corresponding sources.
\end{enumerate*}
%First, per subexpression we count the number of sources to be contacted by summing the cardinalities of all $S$ given by $(SE, S) \in D(P,F)$.
%Second, for each subexpression and all corresponding sources, we count the number additional subexpressions required for an interface-compliant evaluation. 
%\input{tables/example_decompositions}
%Based on these observations, we define the cost of a query decomposition in the following way.
\begin{definition}[Cost of a Query Decomposition]
The cost of evaluating a query decomposition $D(P,F)$ is given by %the sum of the contact cost and the optimal interface-compliant evaluation cost:
$$
cost(D(P,F)) = \sum\limits_{(SE, S) \in D(P, F)} |S| + \sum \limits_{(SE, S) \in D(P, F) \land \forall c \in S} (|\theta^*_{c}(SE)| - 1).
$$
\end{definition}

Note that the proposed query decomposition cost provides a lower bound for evaluating a decomposition while computing the exact cost requires knowledge about the technical configurations of the services in the federation. 
For instance, obtaining solutions from TPF servers might require several requests for paginating the results, while a single request might suffice on a SPARQL endpoint.
%
% Note: Why \theta^* - 1
% We need to subtract one, because we want to assess the additional cost induced for splitting the subquery. 
% If we do not substract it, in the example for F2 D^* would be more expensive that D_1 which does not make sense, since they will both lead to the same query plan. We are already counting the 1 evaluation required by \sum |S| and therefore, do not need to count one additional evaluation 
%
\begin{example}
Let us consider the decomposition $D_2(P, F)$ from \cref{ex:decomp_graphs} and the subexpression $SE_1 =$ $(tp_1\;$ $\textsc{And}\; tp_2\; \textsc{And}\; tp_3)$ to be evaluated at source $S_1 = \{ c_1 \}$.
In contrast to its density, the cost of evaluating $D_2(P, F)$ depends on the LDF interface $c_1$ implements.
If $c_1$ is a SPARQL endpoint, i.e. $int(c_1) = (L_{\textsc{CoreSparql}}, M_{\textsc{Ep}})$, the evaluation $\llbracket SE_1 \rrbracket_{c_1}$ is interface-compliant and thus $|\theta^*_{c_1}(SE_1)| = 1$.
However, if $c_1$ is a TPF server, i.e. $int(c_1) = (L_{\textsc{Tp}}, M_{\textsc{Tpf}})$, the interface-compliant evaluation of $SE_1$ would require evaluating the triple patterns individually with $\theta^*_{c_1}(SE_1) = \llbracket tp_1 \rrbracket_{c_1} \bowtie \llbracket tp_2 \rrbracket_{c_1} \bowtie \llbracket tp_3 \rrbracket_{c_1}$ and thus $|\theta^*_{c_1}(SE_1)| = 3$.
Hence, the evaluation at the TPF server requires two additional subexpressions to be evaluated.
This may lead to higher execution costs as there are potentially more intermediate results to be transferred and the service needs to be contacted at least two additional times.
\end{example}

Finally, we can combine both the density and cost of a query decomposition into the query decomposition problem which aims to obtain a query decomposition that maximizes the expected answer completeness while minimizing the execution cost.

\begin{definition}[Query Decomposition Problem]
Given a BGP $P$ and a federation $F=(C, int, ep)$, the query decomposition problem is finding a query decomposition $D(P,F)$ that minimizes the execution cost while maximizing its density:
$$
D(P, F) = \argmax density(D(P,F)) \land \argmin cost(D(P,F))
$$
\end{definition}

Note that this problem is a multi-objective optimization problem, where there might not be a single best solution but rather a set of optimal trade-off solutions, i.e. Pareto-optimal solutions.

\begin{example}
Consider two alternative example federations which differ in the LDF interfaces of their services $c_1$ and $c_2$:
\begin{itemize}[leftmargin=0.7\parindent]
%\item[] $F_1 = ( \{ c_1, c_2\}, int, ep )$ with $int(c_1) = (L_{\textsc{Tp}}, M_{\textsc{Tpf}})$ and $int(c_2) = (L_{\textsc{CoreSparql}}, M_{\textsc{Ep}})$
\item[] $F_1 = ( \{ c_1, c_2\}, int, ep )$ with $int(c_1) = int(c_2) =$\\ $(L_{\textsc{CoreSparql}}, M_{\textsc{Ep}})$.
%\item[] $F_2 = ( \{ c_1, c_2\}, int, ep )$ with $int(c_1) = (L_{\textsc{CoreSparql}}, M_{\textsc{Ep}})$ and $int(c_2) = (L_{\textsc{Tp}}, M_{\textsc{Tpf}})$
\item[] $F_2 = ( \{ c_1, c_2\}, int, ep )$ with $int(c_1) = (L_{\textsc{Tp}}, M_{\textsc{Tpf}})$ and  $int(c_2) = (L_{\textsc{CoreSparql}}, M_{\textsc{Ep}})$.
\end{itemize}
The density and cost for the query decompositions are given in the following table, where the best values are indicated in bold.

\vspace{1mm}
\begin{footnotesize}
\begin{center}
\begin{tabular}{lcccccccc}
\toprule
& \multicolumn{4}{c}{$F_1$} & \multicolumn{4}{c}{$F_2$} \\
 \cmidrule(lr){2-5}  \cmidrule(lr){6-9}
                        &   $D^*$   &   $D_1$   &   $D_2$   &   $D_3$ &   $D^*$   &   $D_1$   &   $D_2$   &   $D_3$   \\
\midrule
$\sum |S|$                &   5       &   4       &   2       &   2       &   5       &   4       &   2       &   2       \\
$\sum (|\theta^*_c(SQ)| - 1)$   &   0       &   0       &   0       &   0       &   0       &   1       &   1       &   2       \\
% \cmidrule(lr){2-5}  \cmidrule(lr){6-9}
\midrule
$cost$                    &   5       &   4       &   \textbf{2}       &   \textbf{2}       &   5       &   5       &   \textbf{3}       &   4       \\ 
\midrule
$comp$                    &   \textbf{1}       &   \textbf{1}       &   $\frac{9}{11}$     &   $\frac{8}{11}$       &   \textbf{1}       &   \textbf{1}       &   $\frac{9}{11}$       &   $\frac{8}{11}$       \\ 
\bottomrule
\end{tabular}
\end{center}
\end{footnotesize}

\begin{comment}
\begin{table}[t!]
\centering
%\setlength{\tabcolsep}{3pt}
\footnotesize
\caption{Cost and expected answer completeness of the decomposition from \Cref{ex:decomp_graphs} for federations $F_1$ and $F_2$.}
\label{tab:example_costs}
\begin{tabular}{lcccccccc}
\toprule
& \multicolumn{4}{c}{$F_1$} & \multicolumn{4}{c}{$F_2$} \\
 \cmidrule(lr){2-5}  \cmidrule(lr){6-9}
                        &   $D^*$   &   $D_1$   &   $D_2$   &   $D_3$ &   $D^*$   &   $D_1$   &   $D_2$   &   $D_3$   \\
\midrule
$\sum |S|$                &   5       &   4       &   2       &   2       &   5       &   4       &   2       &   2       \\
$\sum (|\theta^*_c(SQ)| - 1)$   &   0       &   0       &   0       &   0       &   0       &   1       &   1       &   2       \\
% \cmidrule(lr){2-5}  \cmidrule(lr){6-9}
\midrule
$cost$                    &   5       &   4       &   2       &   2       &   5       &   5       &   3       &   4       \\ 
\midrule
$comp$                    &   1       &   1       &   $\frac{9}{11}$     &   $\frac{8}{11}$       &   1       &   1       &   $\frac{9}{11}$       &   $\frac{8}{11}$       \\ 
\bottomrule
\end{tabular}
\end{table}
\end{comment}

\vspace{1mm}
\end{example}
The decomposition cost in the example shows how both the number of subexpressions and the number of sources they are evaluated at ($\sum |S|$) as well as the capabilities of the interface $\sum (|\theta^*_c(SE)| - 1)$ have an impact on the overall cost. 
Further, it shows the trade-off between the two conflicting objectives density and cost. % , which is finding a query decomposition that minimizes the cost but at the same time maximizes answer completeness.
In both federations, the decompositions that yield the highest density also have the highest cost and vice versa.
Approaches to solving the query decomposition problem need to determine solutions that yield a suitable (depending on the use case) trade-off between the number of answers and execution cost.
According to \cref{def:decomposition_completeness}, two main factors impact on the density.
First, the triple patterns should be evaluated at as many relevant sources as possible (\ref{item:rule1}).
Second, the more fine-grained the subexpressions for triple patterns that have several relevant sources in common, the higher the density (\ref{item:rule2}).
Similarly, the costs of decompositions originate from two main aspects.
First, contacting fewer sources with larger subexpressions will reduce costs and, second, decomposing the query into subexpressions that are interface-compliant will reduce the cost.
One way of pruning sources without affecting answer completeness is to determine the relevant sources that do not contribute to the final answers of the query \cite{DBLP:conf/esws/SaleemN14}.
However, this can be very challenging for queries with triple patterns that contain terms from common ontologies (e.g., RDF/S, OWL), as they can be answered by many of the sources in the federation.
For this purpose, some approaches rely on pre-computed statistics/catalogues \cite{DBLP:conf/semweb/MontoyaSH17,DBLP:conf/esws/SaleemN14,DBLP:conf/i-semantics/0002PSHN18} and/or the query capabilities of SPARQL endpoints, such as \textsc{Ask} queries \cite{DBLP:journals/tlsdkcs/VidalCAMP16}.  
We propose a query decomposition approach that can be combined with a heuristic-based source pruning method and can be applied for any heterogeneous federation.

\smallbreak
\noindent
\textbf{Query Decomposition Approach.} 
We propose an approach that does not rely on specific statistics about the members and has two central goals:
\begin{enumerate*}[label=(\arabic*)]
    \item maximize the density by evaluating all triple patterns at the relevant sources that contribute to the final answers, and
    \item reduce the execution cost by obtaining subexpressions that leverage the capabilities of the services as much as possible. 
\end{enumerate*}
Furthermore, we add an optional source pruning step to further decrease cost by reducing the number of sources contacted.
%Nonetheless, our approach can be combined with such source pruning approaches, if for instance the necessary statistics are available.
\captionsetup*[subfigure]{position=bottom,textfont=normalfont,labelfont=normalfont}
\begin{figure*}[t]
    \centering
    \subfloat[%\ref{item:step1}: 
    Join Ordering and Union Expressions]{\includegraphics[width=0.22\textwidth]{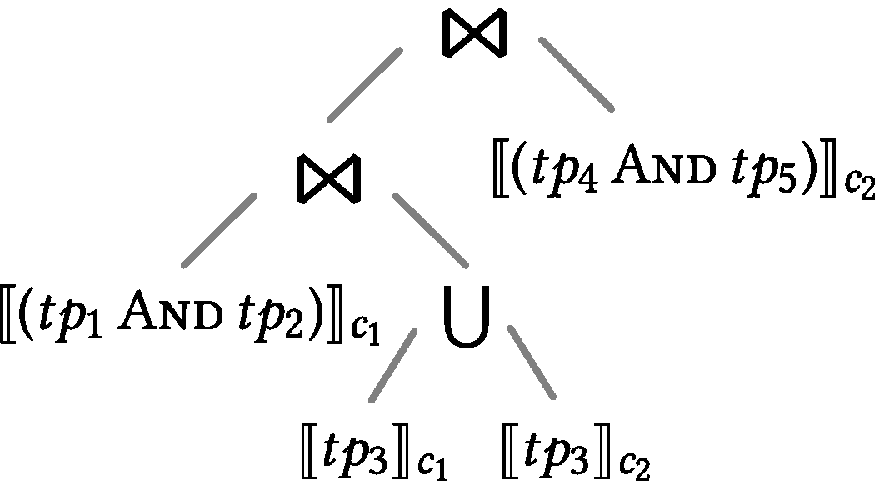}\label{fig:lp_1}}
    %\hfill
    %\hspace{3mm}
    %\subfloat[%\ref{item:step2}: 
    %Unfold Union]{\includegraphics[width=0.22\textwidth]{figures/lp_21.pdf}\label{fig:lp_2}}
    %\hfill
    \hspace{16mm}
    \subfloat[%\ref{item:step3}: 
    Interface-compliant subquery plans]{\includegraphics[width=0.24\textwidth]{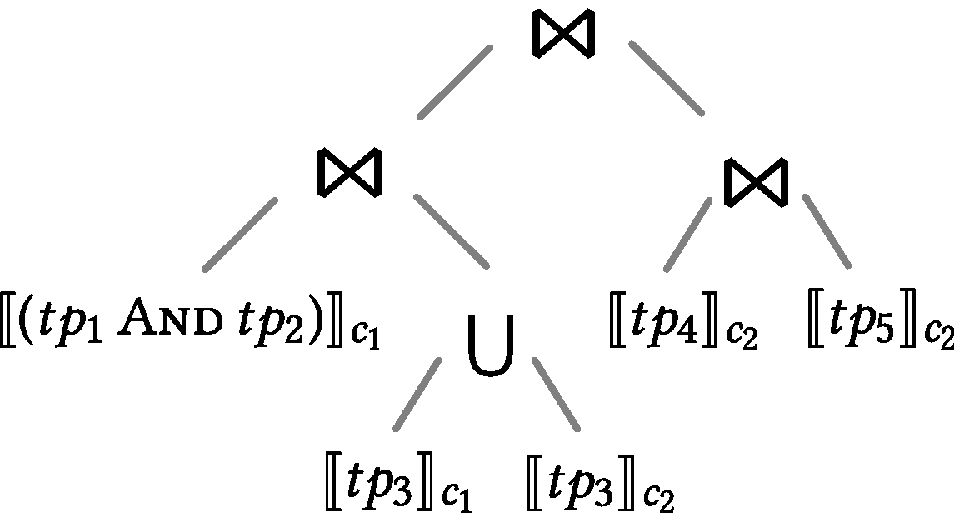}\label{fig:lp_3}}
    %\hfill
    \hspace{16mm}
    \subfloat[%\ref{item:step4}: 
    Place Physical Operators]{\includegraphics[width=0.24\textwidth]{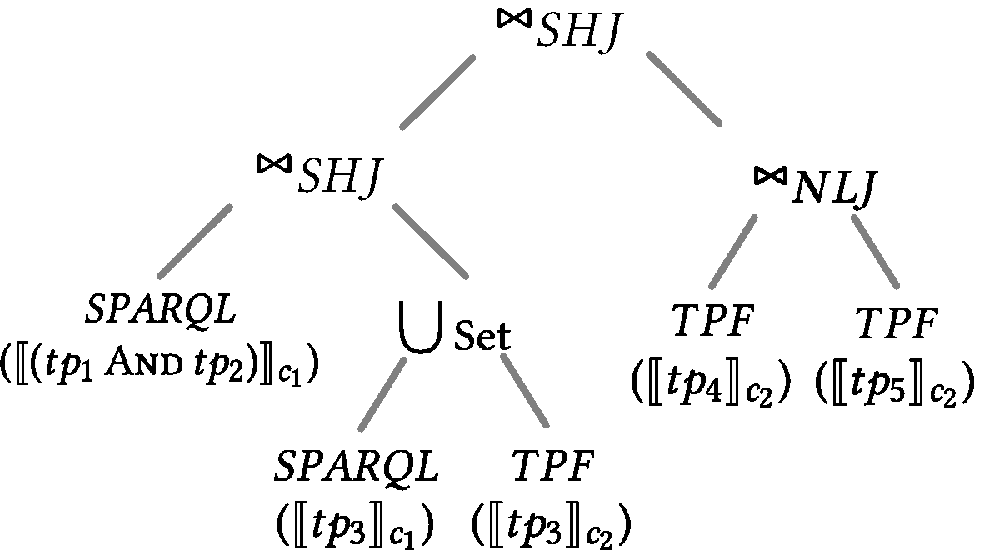}\label{fig:lp_4}}

    \vspace{-0.5\baselineskip}
    \caption{Query planning steps for the query decomposition $D_4(P, F_{ex}) = \{ ((tp_1\, \textnormal{\textsc{And}}\, tp_2), \{c_1\}) , (tp_3, \{c_1, c_2\}), ((tp_4\, \textnormal{\textsc{And}}\, tp_5),\{c_2\})$.}
    \label{fig:example_plan}
\end{figure*}
%% Approach
\begin{algorithm}[t!]
\footnotesize
  \caption{Interface-aware Query Decomposer}
  \label{alg:query_decomposer}
  \SetAlgoLined
  \LinesNumbered
  \KwIn{BGP P = $\{tp_1, \dots, tp_n\}$, Federation $F = (C, int, ep)$}
  \DontPrintSemicolon
  \SetKwRepeat{Do}{do}{while}
  \SetKw{Break}{break}
  \SetKwFunction{relevantSources}{relevantSources}
  \SetKwFunction{pruneSources}{pruneSources}
    $D = \emptyset$\;
    \For{$tp \in P$}{ \label{line:start_dstar}
        $S =$ \relevantSources($tp$)\; \label{line:relevant_sources}%\emptyset$\; 
        $D = D \cup \{ (tp, S)\}$
    } \label{line:end_dstar}

    % Prune Sources
    $D = \pruneSources(D)$\; \label{line:pruning}
    
    \Do{$updated$}{
        $updated = False$\;
        \For{$\forall (SE_i, S_i), (SE_j, S_j) \in D \land SE_i \neq SE_j$}{
            \uIf{$|vars(SE_i) \cap vars(SE_j)| > 0$  \label{line:cond1} \\ 
            $\land |S_i \cup S_j| = 1$  \label{line:cond2}  \\
            $\land (SE_i\; \textsc{And}\; SE_j) \in L_c,\, \forall c \in S_i$ \label{line:cond3} 
            } 
            { \label{line:conditions}
                $D = D \setminus \{ (SE_i, S_i), (SE_j, S_j) \}$\; 
                $D = D \cup \{ ((SE_i\; \textsc{And}\; SE_j), S_i)\}$\;
                $updated = True$\;
                \Break\;
            }
        } 
    }

    \KwRet{$D$}\; \label{line:return_decomposition}
\end{algorithm}
%Our query decomposition approach starts with the complete query decomposition $D^*(P,F)$. 
%It then applies a pruning approach reducing the number sources to be contacted.
%Finally, the decomposer iteratively tries to merge pairs of subexpressions to larger subexpressions if possible.  
The decomposer is outlined in \cref{alg:query_decomposer}.
% atomic Decomposition D star 
Its inputs are a BGP $P$ and a federation $F = (C, int, ep)$.
First, the algorithm creates the atomic decomposition by iterating over each triple patterns $tp$ in $P$, determines the set of relevant sources as $S$, and adds $(tp, S)$ to the decomposition $D$ (\cref{line:start_dstar} - \cref{line:end_dstar}). 
% Pruning
Next, the relevant sources per triple pattern can be pruned in \cref{line:pruning}.
This pruning step is not required, however, it allows for reducing the decomposition cost by 
\begin{enumerate*}[label=\roman*)]
    \item reducing the number of sources to be contacted, and 
    \item allowing to group more triple patterns into subexpressions in the following steps.
\end{enumerate*}
The source pruning approach is interchangeable and we detail our source pruning heuristic in the next paragraph.  
% Merging
After pruning the sources, the algorithm tries to merge as many subexpressions in the decomposition $D$ as possible.
All possible combinations of subexpressions $(SE_i, S_i)$ and $(SE_j, S_j)$ are considered and merged if they fulfill the following three conditions: %in \cref{line:conditions}:
\begin{enumerate}[label=Condition \Roman*,leftmargin=5\parindent]
    \item Both subexpressions have variables in common:\newline $|vars(SE_i) \cap vars(SE_j)| > 0$.  (\cref{line:cond1})
    \item Both subexpressions have exactly one source in common: $|S_i \cup S_j| = 1$. (\cref{line:cond2})
    \item The common source $c$ can evaluate the conjunction of both expressions: $(SE_i\; \textsc{And}\; SE_j) \in L_c$. (\cref{line:cond3})
\end{enumerate}
If two subexpressions fulfill all conditions, the individual subexpressions are removed from the decomposition $D$ and their conjunction is added to $D$.
This process is repeated until no more subexpressions can be merged ($updated = False$). 
A central property of the query decomposition generated by the algorithm is the fact that the evaluation of all subexpressions is compliant with all corresponding sources.
That is, $\forall (SE, S) \in D(P, F): \forall c \in S: \theta_c(SE) = \llbracket SE \rrbracket_c$.
As a result, the interface-compliant evaluation (\cref{def:compliant_decomp_eval}) of all decomposition generate by \cref{alg:query_decomposer} is given as
$$
%\theta_{D(P, F)}(P): =\;  \Bowtie_{(SE_i, S_i) \in D(P, F)}( \cup_{c_j \in S_i}\; \llbracket SE_i \rrbracket_{c_j} ),
\theta_{D(P, F)}(P): =\;  \Bowtie_{(SE_i, S_i) \in D(P, F)}( \cup_{c_j \in S_i}\; \ldqm{SE_i}_{c_j} ).
$$
Note that this property does not require the query planner to find the subexpression minimizing evaluation $\theta_c^*(SE_i)$.
%and it also increases the completeness of the decomposition generated by the approach.

\smallbreak
\noindent
\textbf{Source Pruning Approach.} 
We propose a heuristic that leverages the atomic decomposition graph $G_{D^*(P,F)} = (V^*, E^*)$ and does not rely on data statistics.
Our approach iterates over the source vertices $c_i \in V^*$ by non-increasing out-degree (i.e. starting with the most \emph{popular} source).
For each triple pattern $tp_j$ connected to $c_i$ ($(c_i, tp_j) \in E^*$), the edges to all other sources are removed for $tp_j$: $E^* = E^* \setminus \{(c_k, tp_j) \in E^* \mid \forall c_k \neq c_i \}$.
In addition, the relevant sources for triple patterns with the same common subject are not pruned to maximize completeness.
The rationale for this is the observation that RDF datasets typically follow entity-centric descriptions, where the URI of an entity appears in the subject of triples in the authoritative dataset. 
For example, triples with subject \texttt{dbr:Berlin} are all part of the DBpedia dataset.
%In addition, our pruning approach relies on the assumption that sources typically provide entity-centric data, where entities described in a source appear in the subject of RDF triples, while links to external sources appear in the object. Therefore, the relevant sources for triple patterns about the same subject are not pruned to preserver completeness.

% Old formulation:
%In addition, our approach is based on the assumption that sources typically provide entity-centric data, where the subjects of the triples in a source refer to the same entity. 
%Therefore, the relevant sources for triple patterns about the same subject are not pruned to preserve completeness.

%\begin{example}
%Consider the atomic decomposition graph shown in \cref{fig:d_star}. % of our previous example. 
%The source pruning approach would consider removing the edge $(c_2, tp_3)$ as $tp_3$ can be answered by the most %\emph{popular} source $c_1$. 
%However, as $tp_3$ and $tp_4$ share a common variable in the subject position and $tp_4$ is evaluated at $c_2$, $tp_3$ %will also be evaluated at $c_2$, i.e., the edge $(c_2, tp_3)$ is not removed from the graph.
%\end{example}

%\input{figures/query_planner_example}

\subsection{Query Planner}
\label{sec:query_planner}

The main tasks of the query planner are finding an efficient logical plan and placing physical operators such that the execution time of the query plan is minimized. 
For both tasks, common cost-based query planners leverage statistics on the data of the members in the federation.
In heterogeneous federations, however, the query planning approaches cannot always rely on the same level of statistics from all sources and need to be adjusted to the statistics available at the individual sources.
For instance, obtaining fine-grained statistics might require access to the entire dataset of a source for efficient computation \cite{DBLP:conf/esws/HelingA20} or require the services to be able to execute complex SPARQL expressions, such as aggregate queries.
Furthermore, in the case that the interface language of an LDF service does not support the evaluation of a subexpression from the decomposition, the planner needs to obtain an efficient subplan for evaluating the subexpression over that service.
In this section, we first discuss the steps necessary to obtain efficient query plans, and thereafter, we propose a query planner for query decompositions that respects the interface restrictions in heterogeneous federations.
%
%We begin by discussing the steps of a query planner to obtain a physical query plan for evaluating a decomposition $D(P, F)$ for the BGP $P$ over federation $F$.

\smallbreak
\noindent
\textit{Join Ordering with Union Expressions. }
The query planner determines a join ordering for the subexpressions in a  decomposition that minimizes the number of intermediate results. 
The challenge lies in estimating the size of intermediate results from subexpressions and joins.  
This is particularly difficult in heterogeneous interfaces due to two factors. 
%Especially with heterogeneous interfaces, such estimations can be more difficult due to two factors. 
First, the methods to estimate cardinalities depend on the interface languages and the metadata supported by the interfaces.  
For example, determining the cardinality of a subexpression comprised of two triple patterns could be achieved by a \textsc{Count} query if the interface, e.g. a SPARQL endpoint, supports the evaluation of such expressions. 
However, this could be an expensive operation on the server and thus time-consuming for the client. 
Moreover for other interfaces, such as TPF servers, this would not be possible and the cardinality would need to be estimated according to the metadata of the triple patterns.
If available, statistical data on the data distribution could be used alternatively to estimate the number of intermediate results \cite{DBLP:conf/i-semantics/0002PSHN18,DBLP:conf/i-semantics/CharalambidisTK15}.  
Second, federated plans comprise union operators to combine data from alternative relevant sources.   
In this case, estimating the number of intermediate results from a union operator that will contribute to a join is more difficult due to the different data distributions in each source. 
Therefore, the planner must devise appropriate join orderings in the presence of unions from different sources.  
% computing accurate estimations of join cardinalities in the presence of unions is more difficult as  
%According to the interface-compliant evaluation of a query decomposition (\cref{def:compliant_decomp_eval}), the inner expressions of the joins are union expressions which may need to be evaluated over several sources.
%Therefore, the cardinalities potentially need to be aggregated over the union expressions. \todo{the challenge of the union expressions is not clear from this text}
%Further, the corresponding union operators need to be placed accordingly in the query plan. 
%
\cref{fig:lp_1} shows a join ordering with unions for a query decomposition from the query and federation of our motivating example: $D_4(P, F_{ex}) = \{ ((tp_1\; \textsc{And}\; tp_2), \{c_1\}) , (tp_3, \{c_1, c_2\}), ((tp_4\; \textsc{And}\; tp_5),\{c_2\})$.
%
% Note:
% We could also mention that their might by more advanced physical operators like multi-way hash joins or join operators that can handle several inputs for the same relation (i.e., directly resolve the union) (is there a term for that?)

\smallbreak
\noindent 
\textit{Interface-compliant Subexpression Plans. }        
%
%After the join ordering has been determined on the level of the subexpressions and the union expression have been unfolded, the query planner needs to find interface-compliant subexpression plans.
If a decomposer does not provide decompositions in which the subexpressions $SE$ are interface-compliant, the query planner additionally needs to find subplans that evaluate the subexpression in an interface-compliant manner.
%That is, for those subexpressions $SE$ in a decomposition that are not compliant with the interfaces of one of the sources in $S$ over which $SE$ should be evaluated, i.e., $\exists c \in S: \theta_c(SE) \neq \llbracket SE \rrbracket_c$.
 In those cases, the query planner needs to break down $SE$ into subexpressions that minimize the cost of the interface-compliant evaluation $\theta^*_c(SE)$.
Since the resulting interface-compliant evaluation consists of several joins, the query planner also needs to determine the join ordering for $\theta^*_c(SE)$.
For example, if the service is a TPF server, this would require first splitting the subexpression into its individual triple patterns and thereafter, finding an appropriate join ordering. 
The latter could rely on existing query planning approaches for TPF servers \cite{DBLP:journals/ws/VerborghSHHVMHC16,DBLP:conf/semweb/AcostaV15}.
% NOTE: ADD CROP HERE
\cref{fig:lp_3} shows the interface-compliant evaluation for $\llbracket (tp_4\, \textsc{And}\, tp_5) \rrbracket_{c_2}$ over the DBpedia TPF server ($c_2$) for decomposition $D_4(P, F_{ex})$. 
The evaluation is given by $\theta^*_{c_2}(tp_4\, \textsc{And}\, tp_5) = \llbracket tp_4 \rrbracket_{c_2} \bowtie  \llbracket tp_5 \rrbracket_{c_2}$ and it introduces an additional join operation in the query plan. 

\smallbreak
\noindent 
\textit{Placing Physical Operators. }       
Finally, the query planner selects physical operators to obtain an executable physical query plan. 
This includes placing access operators that retrieve the solution mappings from the services as well as physical join and union operators to process the intermediate results.
The access operators transform the subexpressions into requests that can be processed by the corresponding LDF services. 
Ideally, the access operators leverage the querying capabilities of the interfaces such that the results are obtained efficiently.
For example, traditional federated query engines for  SPARQL endpoints require only access operators that adhere to the SPARQL protocol to get solution mappings from the endpoints.
In heterogeneous federations, however, appropriate access operators for each LDF interface in the federation need to be implemented and placed accordingly by the planner. 
Moreover, physical join operators that implement different join strategies, such as symmetric hash join or bind join, need to be placed effectively as they incur different costs. 
%Besides placing such operators on the level of subexpressions in the decomposition, the query planner also needs to place the join operators in the interface-compliant subexpression plans as well. \todo{I commented the previous sentence because it was not adding much to the message... }
Finally, the planner needs to place the appropriate physical union operators in the plan that respects the semantics of the query language.   
%Specifically, the physical union operator need to be compatible with the assumed semantics.
%For instance, assuming set semantics, the physical union operator would need to filter out duplicate tuples, while duplicates would need to be kept under multi-set (bag) semantics.
\cref{fig:lp_4} shows an example of a physical plan for decomposition $D_4(P, F_{ex})$, where service $c_1$ is a Wikidata SPARQL endpoint and service $c_2$ a the DBpedia TPF server.
\smallbreak
\noindent
\textbf{Query Planning Approach.} 
We now present a heuristic-based query planner for heterogeneous federations.
In particular, it relies on decomposition obtained by \Cref{alg:query_decomposer}.
First, we present the overall planning approach and, thereafter, we present details of our prototypical implementation.
%
% Planning Approach
The query planner is outlined in \cref{alg:query_planner}.
It starts by estimating the cardinality of each subexpression in the decomposition (\cref{line:estimate_cardinality}) and creates a list $L$ in which the subexpressions are sorted by non-decreasing cardinality (\cref{line:sort_q}).
The query planner starts building the query plan with the subexpression $d_1$ with the lowest cardinality and creates the corresponding access plan $T_1$ (\cref{line:access_plan_p1})\footnote{The access plan for $d_i = (SE_i,S_i)$ refers to the union of evaluating subexpression $SE_i$ at each source in $S_i$.}.
It iterates over the remaining subexpressions in $L$ and determines the next subexpression to join $T_1$ with. 
This is either a remaining subexpression with the lowest cardinality and a common variable (\cref{line:join_q2}) or if there is no join remaining in the BGP, it is the subexpression with the lowest cardinality (\cref{line:next_q2}).
Once the subexpression $d_2$ is selected, the access plan $T_2$ for $d_2$ is created (\cref{line:access_plan_p2}) and the appropriate physical join operator $O$ is determined (\cref{line:get_physical_operator}).
Finally, $T_1$ becomes the \texttt{JoinPlan} of $T_1$ and $T_2$ (\cref{line:join_p1_p2}).
When $L$ is empty, the final plan $T_1$ is returned (\cref{line:return_plan}).
% The join plan produces the Cartesian product in case there is no common join variable
%
\begin{algorithm}[t!]
\footnotesize
  \caption{Query Planning Algorithm}
  \label{alg:query_planner}
  \SetAlgoLined
  \LinesNumbered
  \KwIn{Decomposition $D(P,F)= \{ (SE_1, S), \dots , (SE_n, S_n)  \}$}
  \DontPrintSemicolon
  \SetKwFunction{estimateCardinality}{estimateCardinality}
  \SetKwFunction{sort}{sort}
    \SetKw{Break}{break}
  \SetKwFunction{getPhysicalOperator}{getPhysicalOperator}
  \SetKwFunction{AccessPlan}{AccessPlan}
  \SetKwFunction{JoinPlan}{JoinPlan}
 
    %\tcp{Find a robust plan if the cheapest plan is not robust enough}
    List $L$\;
    \For{$(SE_i, S_i) \in D(P, F)$}{
        $card_i = \estimateCardinality(SE_i, S_i)$\; \label{line:estimate_cardinality}
        $L.append((SE_i, S_i, card_i)) $\;
    }
    $L = \sort(L, card_i)$ \tcp*[h]{Sort $L$ by non-decreasing $card_i$}\; \label{line:sort_q}
    
    $d_1 = L.get(0)$\;
    $L.remove(0)$\;
    $T_1 = \AccessPlan(d_1)$\; \label{line:access_plan_p1}
    
    \While{$|L| > 0$}{
        $d_2 = L.get(0)$\; \label{line:next_q2}
        \For{$i = 1; i < |L|; i++$}{
            $(SE_i, S_i, card_i) = L.get(i)$\;
            \uIf{$|vars(T_1) \cap vars(SE_i)| > 0$}{
                $d_2 = (SE_i, S_i, card_i)$\; \label{line:join_q2}
                $L.remove(i)$\;
                \Break \;
            }
        }
        $T_2 = \AccessPlan(d_2)$\; \label{line:access_plan_p2}
        $O = \getPhysicalOperator(T_1, T_2)$\; \label{line:get_physical_operator}
        $T_1 = \JoinPlan(T_1, T_2, O)$\;  \label{line:join_p1_p2}
    }
    
    \KwRet{$T_1$}\; \label{line:return_plan}
\end{algorithm}
%
% Implementation Details
After presenting the generic planning approach, we now provide details on the specific steps in our prototypical implementation.
The current implementation focuses on the three well-known LDF interfaces: Triple Pattern Fragments (TPF), Bindings-Restricted Triple Pattern Fragments (brTPF), and SPARQL endpoints. 
Further, it relies on the properties of decompositions generated by our interface-aware query decomposer presented in \cref{alg:query_decomposer}. 
That is, each subexpression $SE_i$ is interface-compliant for all sources in $S_i$.
%
%First, \texttt{estimateCardinality} estimates the cardinality for each subexpression in in the following way.
For each service $c \in S_i$, \texttt{estimateCardinality} (\cref{line:estimate_cardinality}) obtains the estimated cardinality $card_i^c$ for the subexpression $SE_i$ at the service $c$ in line with the interface language and the metadata of $c$.
As evaluating $SE_i$ at several sources reflects a union operation, it then sums up those individual cardinalities to obtain the total cardinality of $SE_i$ at all sources: $card_i = \sum_{c \in S_i} card^c_i$. 
If $SE_i$ is a triple pattern and the source is a brTPF or a TPF server, we request the triple pattern and use the \texttt{void:count} in the metadata as the cardinality estimation. 
If $SE_i$ is a BGP or a triple pattern and the source is a SPARQL endpoint, we use a \textsc{Count} query to estimate the cardinality. 
Further, we estimate the join cardinality of two subexpressions $SE_i$ and $SE_j$ as the minimum of their cardinalities. 
Next, we implement appropriate access operators for all three interfaces.
Since all subexpressions are compliant with the interface, we do not need to first obtain an interface-compliant evaluation in the \texttt{AccessPlan}s.
Finally, we determine the physical join operator according to the estimated number of requests to execute the join.
We distinguish between two different common join strategies: symmetric hash join and bind join.
The reason to use the number of requests to determine the join strategy is two-fold:
\begin{enumerate*}[label=\roman*)]
    \item the number of requests directly have an effect on the execution time, and
    \item fewer requests lead to a reduced load on the services in the federation.
\end{enumerate*}
Thus, we compare the number of requests necessary when placing a bind join or a symmetric hash join and choose the operator that yields fewer requests.
The request estimations depend on the implementation of the physical join operator, which we detail in the following section.

\subsection{Physical Operators}
\label{sec:polymorphic_join_operator}
%In heterogeneous federations, the  
% Access operators
% Join operators
The heterogeneity of LDF interfaces in a federation introduces challenges but also opens opportunities for implementing novel physical operators. 
Access operators to retrieve answers from LDF services need to be implemented in efficient ways reducing the load on the LDF services and the time for obtaining results to improve query execution time.
For instance, TPF servers have a \emph{page size} configuration that limits the number of answers that are returned upon a requested triple pattern.
%Furthermore, as TPF server return triples and not solution mappings, the access operators needs to map the obtained triples to corresponding solution mappings.
%On the other hand, SPARQL endpoints typically do not impose any restriction on the expressivity of the queries that can be evaluated at the endpoints. \todo{<- TRUE but confusing for the reader ... I will comment this out}
Additionally, many public SPARQL endpoints are configured with fair use policies that can lead to zero or incomplete query results \cite{DBLP:conf/semweb/SouletS19}.
Consequently, implementations of access operators for SPARQL endpoints should not overload the SPARQL endpoints and adhere to the usage policies.
Yet, physical join operators can be designed to simultaneously handle different LDF interfaces and follow different join strategies depending on the capabilities of the underlying services. 
We call these kinds of operators \textit{polymorphic} and present a novel Polymorphic Bind Join tailored to TPF, brTPF, and SPARQL interfaces.  
%For instance, bind join strategies can leverage the capabilities of the LDF services.
%Further properties, such as the sorted triple pattern in TPF server with an HDT backend also provides potential for improving the implementation of physical join operators.
%
%An additional consideration regarding both access and join operators would be the decision on blocking or non-blocking operators which will have on the diefficiency of the approach.

\smallbreak
\noindent
\textbf{Polymorphic Bind Join.}
%We propose the concept of \textit{Polymorphic Physical Operators} adapt to the specifics of various LDF interfaces to leverage the capabilities of different interfaces. 
%We propose a physical join operator called the Polymorphic Bind Join (PBJ) whose implementation is tailored to leverage the interfaces' capabilities. %TPF/brTPF servers and SPARQL endpoints. 
%We propose the Polymorphic Bind Join (PBJ) operator that implements an interface-compliant bind join strategy.
The Polymorphic Bind Join (PBJ) implements a Nested Loop Join algorithm that is able to adjust its join strategy according to the LDF interface.
It simultaneously executes a tuple- and block-based nested loop join according to the supported interface language. %allow for provided several solution mappings.
Our current implementation supports the languages $L_{\textsc{Tp}}$, $L_{\textsc{Tp+Values}}$ and $L_{\textsc{CoreSparql}}$.
By leveraging the capabilities of each service, PBJ reduces the number of requests when accessing more capable sources using the block-based approach.
In particular, PBJ is designed for cases where the inner relation is either an access operator or the union of access operators. 
For each LDF interface $f$, a block size $B_f$ is defined.
During the execution, the operator keeps a reservoir per service that is filled by tuples from the outer relation.
When the reservoir reaches the block size $B_f$ of the corresponding LDF interface, the bindings from the reservoir are requested at the services.
For example, when querying a TPF server in a nested loop join, each solution mapping of the outer relation is used to instantiate and resolve the triple pattern of the inner relation, hence,  $B_{\textsc{Tpf}} = 1$.
However, as the interface languages of brTPF servers and SPARQL endpoints support SPARQL values expressions, the PBJ changes its operation accordingly by requesting a triple pattern or a subexpression with several bindings.
The number of bindings that can be sent to a brTPF server $B_{\textsc{brTpf}}$ depends on the server configuration \cite{DBLP:journals/corr/HartigA16}.
For SPARQL endpoints, $B_{\textsc{Ep}}$ is not limited, yet too many values may lead to long runtimes at the endpoint and potentially incomplete results.\footnote{In our implementation, we set $B_{\textsc{brTpf}} = 30$ and $B_{\textsc{Ep}} = 50$, to reduce the requests while not overloading the endpoint.}
%Hence, we set $B_{\textsc{Ep}} = 50$ in our implementation to reduce the requests while not overloading the endpoint.
%Further, smaller block sizes improve the continuous production of solutions due to the semi-blocking nature of the PBJ.

The proposed query planner selects a Symmetric Hash Join (SHJ) or Polymorphic Bind Join (PBJ) in \texttt{getPhysicalOperator} (\cref{line:get_physical_operator}) depending on the estimated number of requests.
The number of requests to execute the SHJ or PBJ depends on the sub-plans $T_1$ and $T_2$. 
If $T_1$ is an \texttt{AccessPlan}, the number of requests to obtain the tuples of $T_1$ are determined by its cardinality $card_{T_1}$ and the interfaces over which $T_1$ is evaluated.
Otherwise, if $T_1$ is a \texttt{JoinPlan}, no additional requests are necessary to obtain the tuples for $T_1$.
For the first case, the requests $R_{acc}(T_1)$ depend on the maximum number of tuples that can be obtained per requests from the corresponding LDF service, which we denote as $Max_{\textsc{Ep}}$, $Max_{\textsc{brTpf}}$, and $Max_{\textsc{Tpf}}$.\footnote{In our implementation, we set $Max_{\textsc{brTpf}} = 100$~\cite{DBLP:journals/corr/HartigA16} and $Max_{\textsc{Tpf}} =100$~\cite{DBLP:journals/ws/VerborghSHHVMHC16}, and $Max_{\textsc{Ep}} = 10000$ (most common value reported at \url{https://sparqles.ai.wu.ac.at/}).}

\begin{equation*}
\begin{aligned}
\scriptstyle
R_{acc}(T) = 
\sum \limits_{\substack{c \in S \land \\ int(c) = \textsc{Ep}}}
\left\lceil \frac{card_T^c}{Max_{\textsc{Ep}}} \right\rceil +  
\sum \limits_{\substack{c \in S \land \\ int(c) = \textsc{brTpf}}}
\left\lceil \frac{card_T^c}{Max_{\textsc{brTpf}}} \right\rceil +  
\sum \limits_{\substack{c \in S \land \\ int(c) = \textsc{Tpf}}}
\left\lceil \frac{card_T^c}{Max_{\textsc{Tpf}}} \right\rceil  
\end{aligned}
\end{equation*}
As a result, we can compute the number of request for the SHJ as the sum of the requests for the two sub-plans:
$$
R_{SHJ}(T_1, T_2) = R_{acc}(T_1) + R_{acc}(T_2)
$$
For the PBJ, we need to determine the number of requests that need to be performed in the inner relation $R_{bind}(T_1, T_2)$ , which depends on the cardinality $card_{T_1}$ of the outer relation $T_1$ and the block sizes for the services in the inner relation:
\begin{equation*}
\begin{aligned}
\scriptstyle
R_{bind}(T_1, T_2) = 
\sum \limits_{\substack{c \in S_2 \land \\ int(c) = \textsc{Ep}}} 
\left\lceil \frac{card_{T_1}}{B_{\textsc{Ep}}} \right\rceil +  
\sum \limits_{\substack{c \in S_2 \land \\ int(c) = \textsc{brTpf}}}
\left\lceil \frac{card_{T_1}}{B_{\textsc{brTpf}}} \right\rceil +  
\sum \limits_{\substack{c \in S_2 \land \\ int(c) = \textsc{Tpf}}} 
\left\lceil \frac{card_{T_1}}{B_{\textsc{Tpf}}} \right\rceil  
\end{aligned}
\end{equation*}
The overall number of requests for the PBJ is 
$$
R_{PBJ}(T_1, T_2) = R_{acc}(T_1) + R_{bind}(T_1, T_2).
$$

\section{Experimental Evaluation}
\label{sec:experimental_evaluation}
\captionsetup*[subfigure]{position=bottom,textfont=normalfont,labelfont=normalfont}
\begin{figure*}[t]
    \centering
    \hfil
    \subfloat[\fedi]{\includegraphics[width=0.508\textwidth]{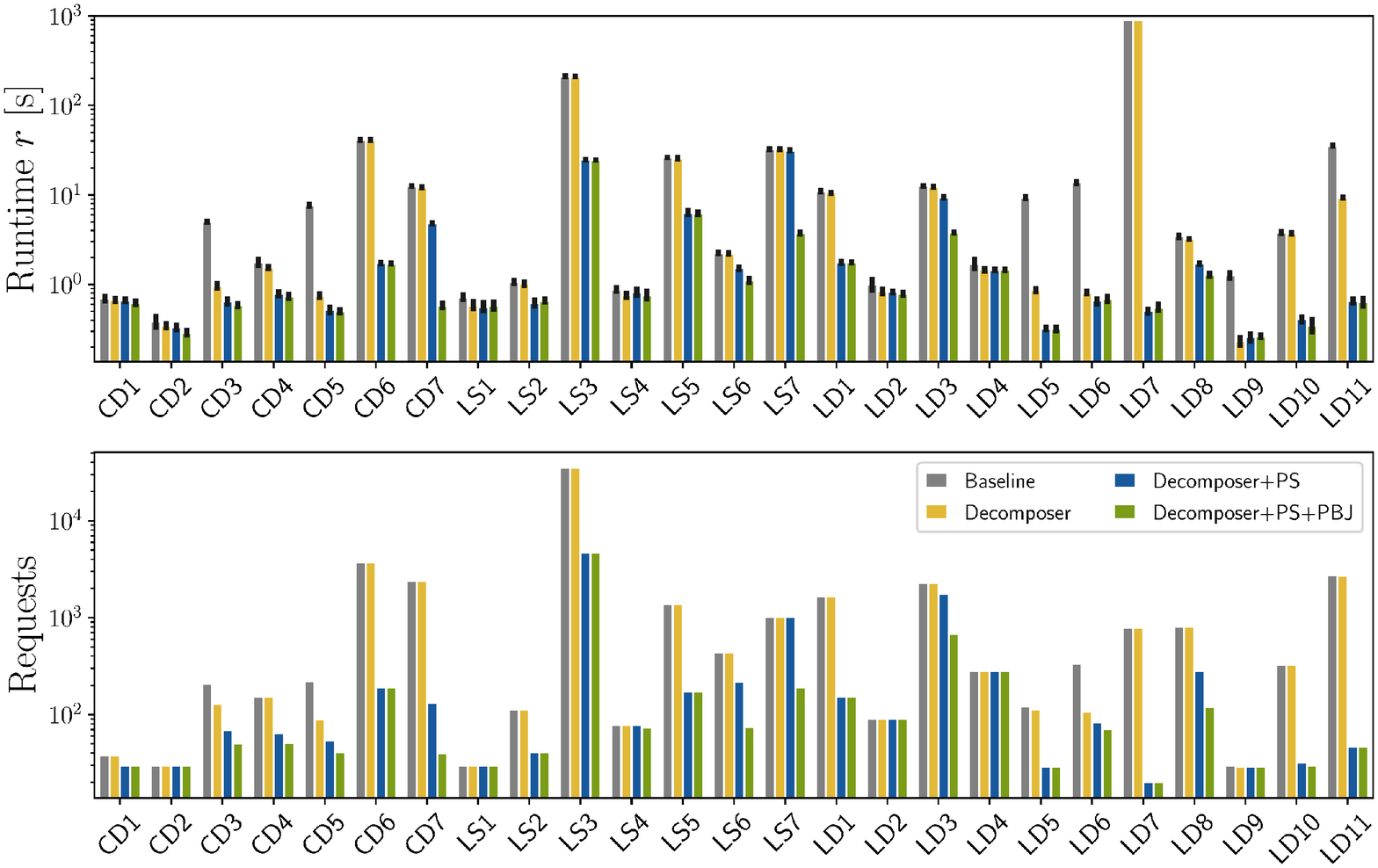}\label{fig:fed1_rr}}
    \hfill
    %\hspace{16mm}
    \subfloat[\fedii]{\includegraphics[width=0.492\textwidth]{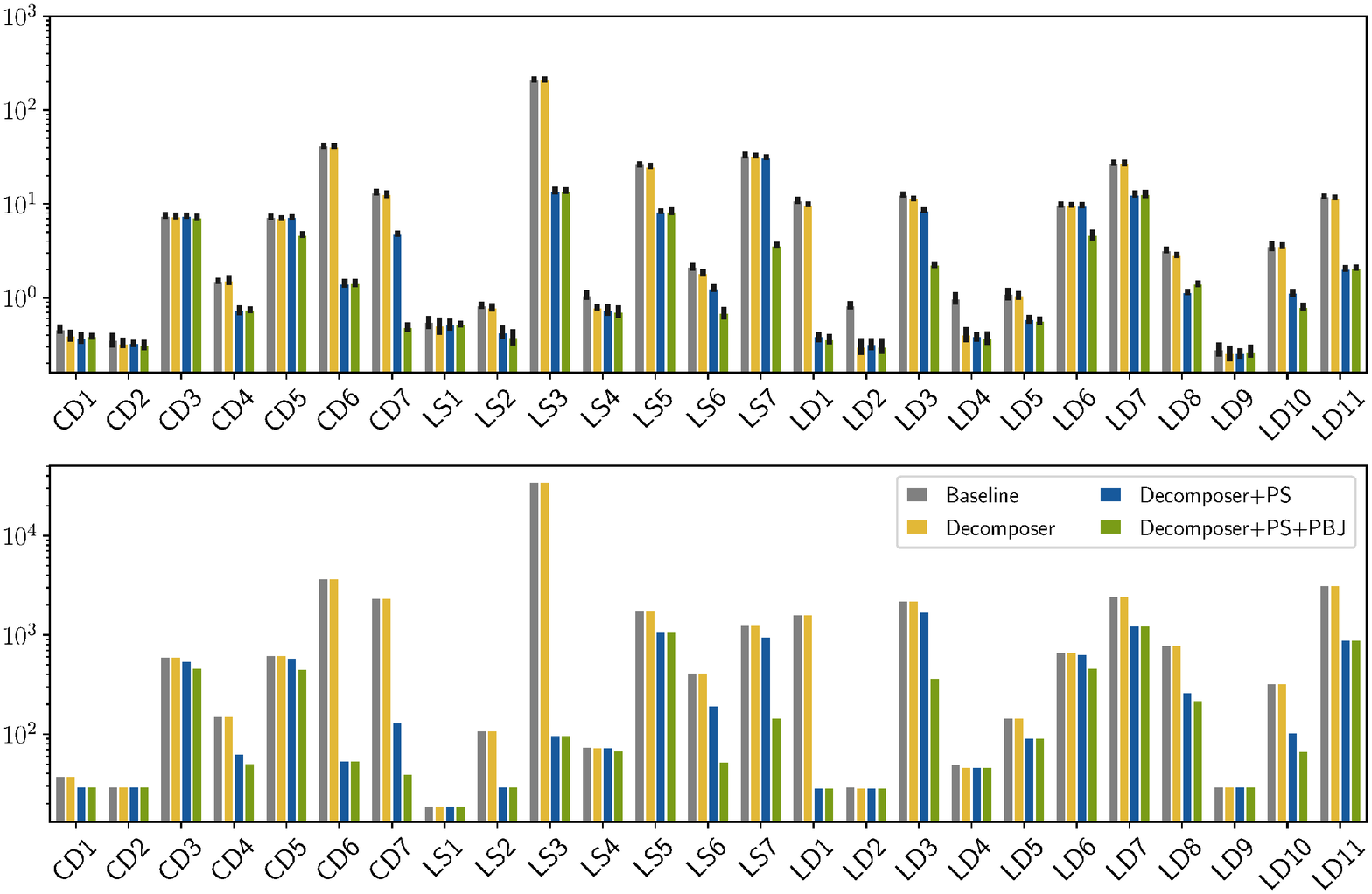}\label{fig:fed2_rr}}
    \vspace{-0.5\baselineskip}
    \caption{Average runtimes [s] (log-scale) and total number of requests (log-scale) for each query and both federations.}
    \label{fig:runtime_requests}
\end{figure*}

We evaluate a prototypical implementation of the interface-compliant query decomposer, query planner, and polymorphic bind join.
The goal is to investigate the impact of the components on the performance when querying heterogeneous federations of LDF interfaces.
%when exploiting the capabilities of the interfaces according to the components of our framework.

% Datasets and queries
\smallbreak
\noindent\textbf{Datasets and Queries.}
We use the well-known FedBench benchmark~\cite{DBLP:conf/semweb/SchmidtGHLST11} which is comprised of $9$ datasets and tailored to assess the performance of federated SPARQL querying strategies. 
We use a total of $25$ queries from Cross Domain (\textsf{CD1-7}), Life Science (\textsf{LS1-7}) and Linked Data (\textsf{LD1-11}) in our evaluation.

\smallbreak
\noindent\textbf{Federations.}
We evaluate our approach on two heterogeneous federations \fedi and \fedii shown in \cref{tab:federations} to study the performance in different scenarios.
The central difference between the federations is that in \fedi the three largest datasets are accessible via SPARQL endpoints while in \fedii they are accessible via TPF servers. 
The other datasets are accessible via TPF or brTPF servers.

\begin{table}[t!]
\centering
\setlength{\tabcolsep}{1pt}
\scriptsize
\caption{Heterogeneous Federations: \fedi, \fedii. SPARQL endpoints indicated in bold and brTPF servers in italic.}
\label{tab:federations}
\begin{tabular}{llllllllll}
\toprule
 & DBpedia & NYTimes & LinkedMDB & Jamendo & GeoNames & SWDF & KEGG & Drugbank & ChEBI \\
\midrule
\fedi & \textbf{\textsc{Sparql}} & \emph{\textsc{brTpf}} & \emph{\textsc{brTpf}} & \textsc{Tpf} & \textbf{\textsc{Sparql}}   & \textsc{Tpf} & \emph{\textsc{brTpf}} & \textsc{Tpf} & \textbf{\textsc{Sparql}}   \\
\midrule
\fedii & \textsc{Tpf} & \emph{\textsc{brTpf}} & \emph{\textsc{brTpf}} & \textbf{\textsc{Sparql}}  & \textsc{Tpf} & \textbf{\textsc{Sparql}}  & \emph{\textsc{brTpf}} & \textbf{\textsc{Sparql}}  & \textsc{Tpf} \\
\bottomrule
\end{tabular}
\end{table}

\smallbreak
\noindent\textbf{Implementation.}
We implemented a prototypical federated query engine for heterogeneous federations that implements the proposed query planner, decomposer, source pruning (\textsf{PS}), and polymorphic bind join (\textsf{PBJ}) operator. 
Our implementation is based on CROP \cite{DBLP:conf/semweb/HelingA20} and implemented in Python 2.7.13. 
%e\footnote{\url{https://github.com/Lars-H/crop}}
%We implemented new physical access operators for SPARQL endpoints and brTPF server.
As \textsf{Baseline}, we use execute the query plans from our query planner for the atomic decompositions.
The decomposer, source pruning, and PBJ are disabled in the \textsf{Baseline}.
%That is, decomposer, source pruning and interface-awareness for the PBJ by setting the block sizes $B_{\textsc{Tpf}} = B_{\textsc{brTpf}} = B_{\textsc{Ep}} = 1$. 
Note that, while Comunica \cite{DBLP:conf/semweb/TaelmanHSV18} can query heterogeneous interfaces, its performance is currently not competitive as it does not implement query decomposition, source pruning, or polymorphic join operators. Therefore, we do not consider Comunica in our evaluation.
We use the \texttt{Server.js} v2.2.3\footnote{\url{https://github.com/LinkedDataFragments/Server.js}} and original Java brTPF server implementation \cite{DBLP:journals/corr/HartigA16} to deploy the TPF and brTPF servers with HDT \cite{DBLP:journals/ws/FernandezMGPA13} backends.
% brTPF URL: http://olafhartig.de/brTPF-ODBASE2016/
We used Virtuoso v07.20.3229 with the default \texttt{virtuoso.ini} (cf. supplemental material).
%(provided in the supplemental material to setup the SPARQL endpoints. %\textsuperscript{\ref{fn:supp_material}})
All LDF services and the client were executed on a single Debian Jessie server (2x16 core Intel(R) Xeon(R) E5-2670 2.60GHz CPU; 256GB RAM) to avoid network latency.
The timeout was set to 900 seconds. 
After a warm-up run, the queries were executed five times.
The source code, experimental results, and additional material are provided in the supplemental material of this submission.
%during reviewing in a DropBox folder online.\footnote{\scriptsize{\url{https://www.dropbox.com/sh/tcl5ae7kuszmf49/AABwe1tWOpW_-rCOtOSAkOa0a?dl=0}}\label{fn:supp_material}}

% Why do we not compare to other approaches: Hybrid federations / exitsing approaches cannot be applied , Future work could focus on how things could be transfered from "normal" federations to hybrid federations

% Metrics:
\smallbreak
\noindent\textbf{Metrics.}
We evaluated the performance by the following metrics: 
\begin{enumerate*}[label=(\roman*)]
    \item \textit{Runtime}: Elapsed time spent by the engine evaluating a query.
    \item \textit{Number of Requests}: Total number of requests submitted to the LDF services during the query execution.
    \item \textit{Number of Answers}: Total number of answers produced. 
    \item \textit{Diefficiency}: Continuous efficiency as the answers are produced over time \cite{DBLP:conf/semweb/AcostaVS17}.
\end{enumerate*}
\begin{table}[t!]
\centering
\setlength{\tabcolsep}{3.5pt}
\footnotesize
\caption{Average total runtime $\sum r$, number of requests $\sum req.$, and answers $\sum ans.$ per run as well as the mean decomposition completeness $\overline{comp}$ and decomposition cost $\overline{cost}$\textsuperscript{\ref{fn:normalized_cost}}.}% for the different approaches and federations.}
\label{tab:evaluation_table}
\begin{tabular}{llrrrrr}

\toprule
&    &      $\sum r$ &      $\sum req.$ &       $\sum ans.$ & $\overline{comp}$ & $\overline{cost}$  \\
\midrule
\multirow{4}{*}{\begin{sideways}\fedi \end{sideways}} & \textsf{Baseline}           &         1337.87 &          54452 &  \textbf{13534} &  \textbf{1.0} &                1.0 \\
& \textsf{Decomposer}        &         1274.15 &          53958 &  \textbf{13534} &  \textbf{1.0} &               0.95 \\
& \textsf{Decomposer+PS}     &           93.66 &           9645 &           13171 &          0.77 &      \textbf{0.55} \\
& \textsf{Decomposer+PS+PBJ} &  \textbf{54.69} &  \textbf{7271} &           13171 &          0.77 &      \textbf{0.55} \\
\midrule
\multirow{4}{*}{\begin{sideways}\fedii \end{sideways}} & \textsf{Baseline}            &          433.37 &          57671 &  \textbf{13578} &  \textbf{1.0} &                1.0 \\
& \textsf{Decomposer}        &          425.15 &          57662 &  \textbf{13578} &  \textbf{1.0} &                0.9 \\
& \textsf{Decomposer+PS}     &          116.45 &           9040 &           13171 &          0.77 &      \textbf{0.53} \\
& \textsf{Decomposer+PS+PBJ} &  \textbf{69.38} &  \textbf{6121} &           13171 &          0.77 &      \textbf{0.53} \\
\bottomrule
\end{tabular}
\end{table}

\subsection{Experimental Results}
\label{sec:experimental_results}
% Q1: How is the runtime impacted? Baseline vs. Complete Approach 
% Q2: How is the number of requests impacted? Baseline vs. Complete Approach
%
We start providing an overview of the performance of the different components.
In \cref{fig:fed1_rr} and \cref{fig:fed2_rr} the mean runtimes and number of requests are shown per query for \fedi and \fedii.
The values are also summarized in \cref{tab:evaluation_table}.
Considering the impact of the individual components, the results show that enabling the decomposer without pruning the sources and no PBJ (\textsf{Decomposer}), only provides a slight improvement in the runtime over the \textsf{Baseline}, even though all queries yield the same number of requests or less.
This is because, without source pruning, only exclusive groups can be merged by the decomposer.
The results when adding the source pruning approach (\textsf{Decomposer+PS}) show that pruning sources considerably reduces both the runtime and the number of requests for the majority of queries.
The reasons for the improvement are two-fold: 
\begin{enumerate*}[label=\roman*)]
    \item the decomposer can create more and larger subexpressions, and 
    \item fewer services are contacted during the execution of the query plan. 
\end{enumerate*}
Finally, with the polymorphic join operator (\textsf{Decomposer+PS+PBJ}), we observe the lowest overall runtimes and number of requests in both federations.
In \fedi, executing all queries with \textsf{Decomposer+PS+PBJ} is more than $34$ times faster than the \textsf{Baseline} and $6$ times faster in \fedii.
The results show that our interface-aware federated query approaches, that adjust to the specifics of heterogeneous interfaces, can greatly improve the performance in terms of runtime.
Simultaneously, it reduces the load on the servers by requiring fewer requests.  
The results show that the interfaces present in the federation (\fedi vs. \fedii) substantially impact the querying performance when not considering the interfaces' capabilities (\textsf{Baseline}).
Yet, our interface-aware solution (\textsf{Decomposer+PS+PBJ}) enables similar performance results regardless of the interfaces.
%However, when implementing the interface-aware solution (\textsf{Decomposer+PS+PBJ}) the performance is similar between the federations.
% Optimization time impact?

% Q3: How are cost and decomposition completeness impacted by the different appraoches? % Baseline, Decomposer (No Pruning), Decopmoser with Pruning
% Show that the lowest cost are achieved with Decomposer+Pruning
% --> Table
\smallbreak
\noindent\textbf{Query Decomposition.}
The results show the effectiveness of the proposed $density$ measure as a proxy for completeness and $cost$ measures as means to assess the expected execution cost of query decompositions. %of our framework. 
In \cref{tab:evaluation_table}, we can observe that, in both federations, the decomposer without source pruning yields complete answers with $\overline{density} = 1.0$, since only exclusive groups are merged (\ref{item:rule3}).
The cost can only be slightly reduced (\fedi: $\overline{cost} = 0.95$ and \fedii: $\overline{cost} = 0.9)$\footnote{Cost values are normalized: $cost(D(P,F)) /cost(D^*(P,F))$\label{fn:normalized_cost}}.
However, adding the source pruning (\textsf{Decomposer+PS}) enables decompositions with about half the cost.
Contacting fewer services reduces the cost but also leads to a reduction in the expected completeness ($\overline{density} = 0.77$) and to fewer answers ($\sum ans$) that are being produced. 
$97\%$ of all answers are still produced when sources are pruned.\footnote{In \fedi, the \textsf{Baseline} does not yield all answers, due to a timeout in \textsf{LD7}.}
These results show that the improvement achieved by the decomposer in its ability to leverage the interfaces' capabilities depends on the source pruning.
%For FedBench, our source pruning heuristic enables substantial benefits.
%However, as the impact on answer completeness may be higher in other federations, future research should investigate more sophisticated pruning approaches for heterogeneous federations.
%However, as the impact on answer completeness may be higher in other federations, more sophisticated pruning approaches for heterogeneous federations should be investigated in future research. 

\captionsetup*[subfigure]{position=bottom,textfont=normalfont,labelfont=normalfont}
\begin{figure*}
\subfloat[\fedi: \textsf{LS3}]{\includegraphics[width=0.23\textwidth]{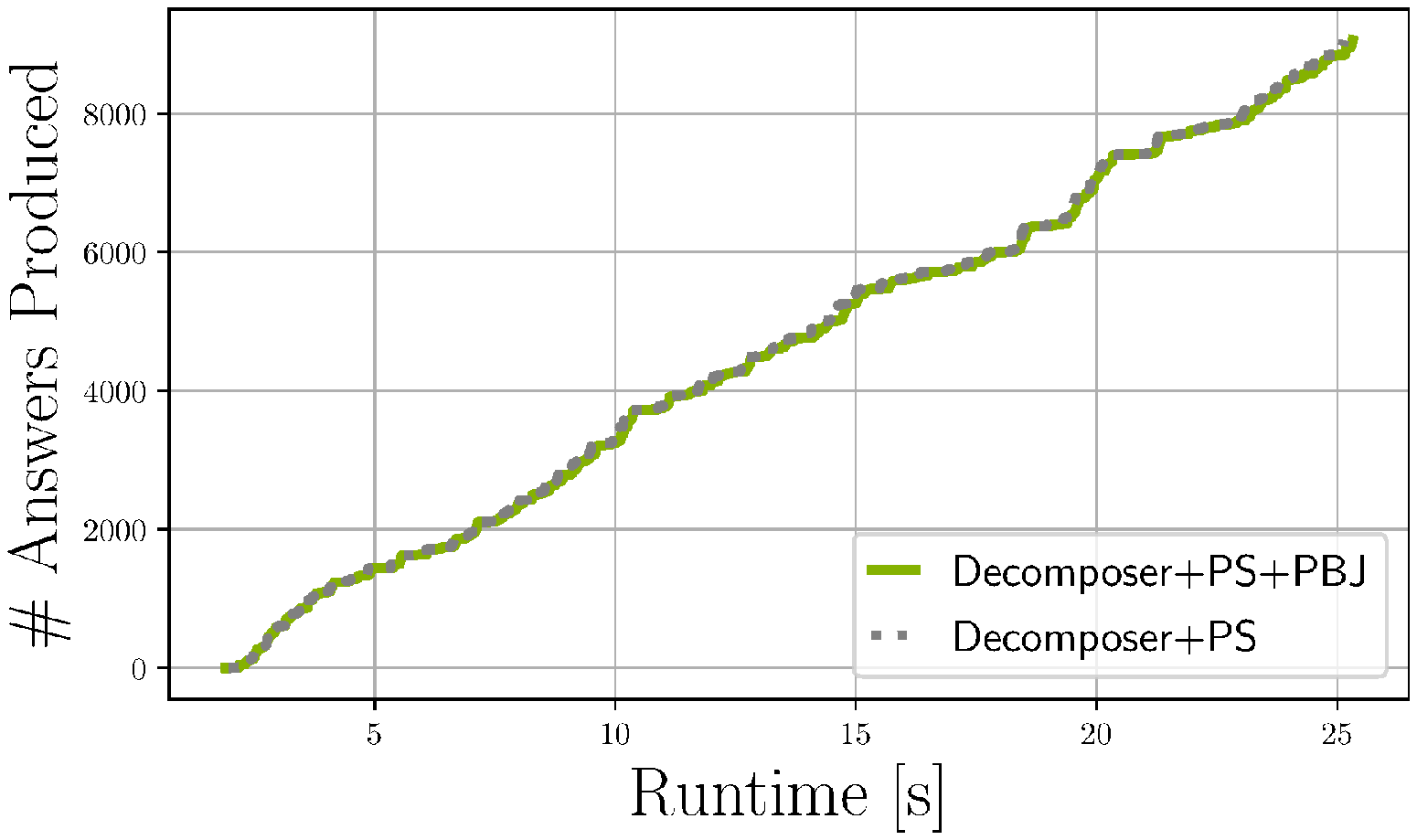}\label{fig:dief_ls3}}
\subfloat[\fedi: \textsf{LD3}]{\includegraphics[width=0.23\textwidth]{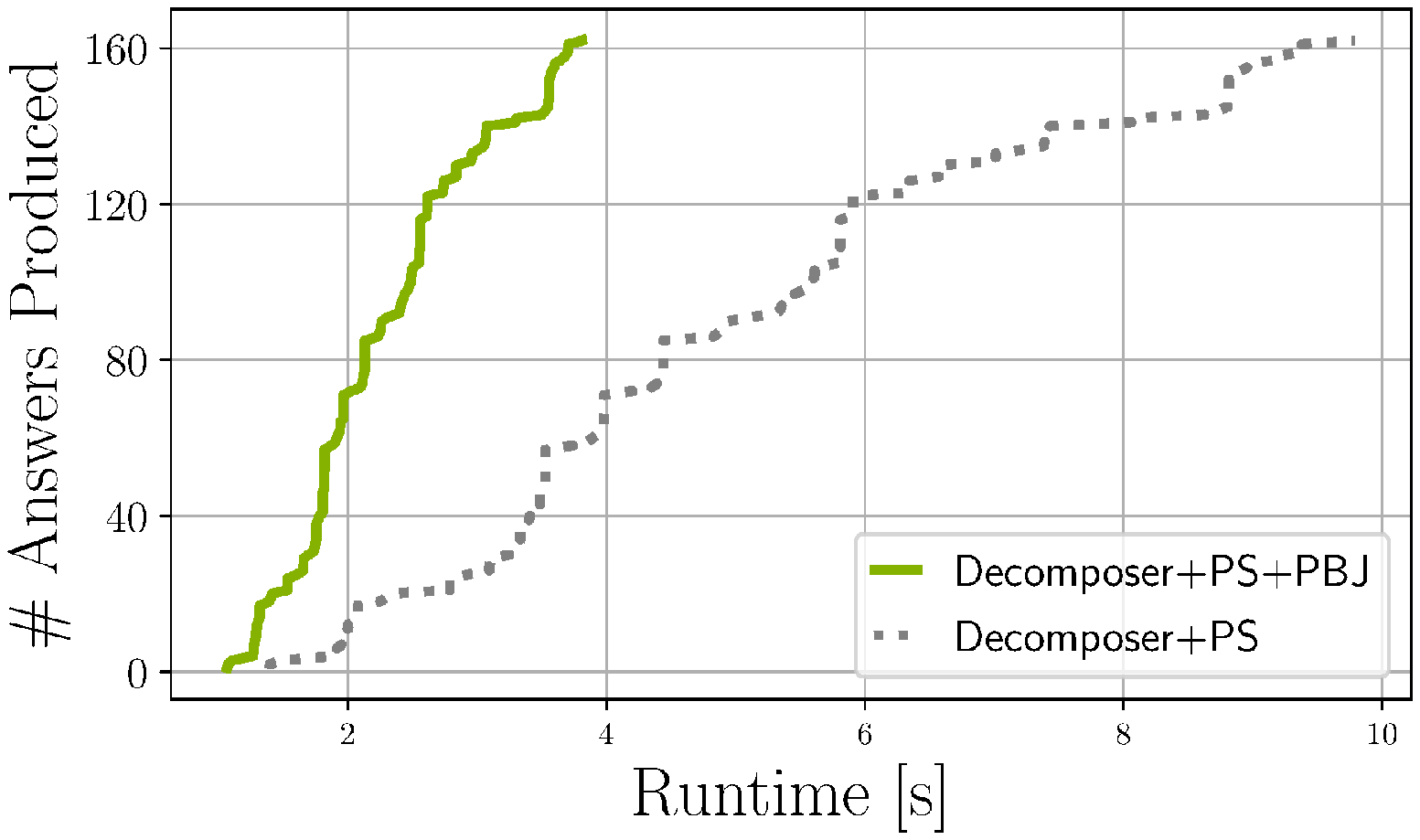}\label{fig:dief_ld3}}
\subfloat[\fedii: \textsf{LS6}]{\includegraphics[width=0.23\textwidth]{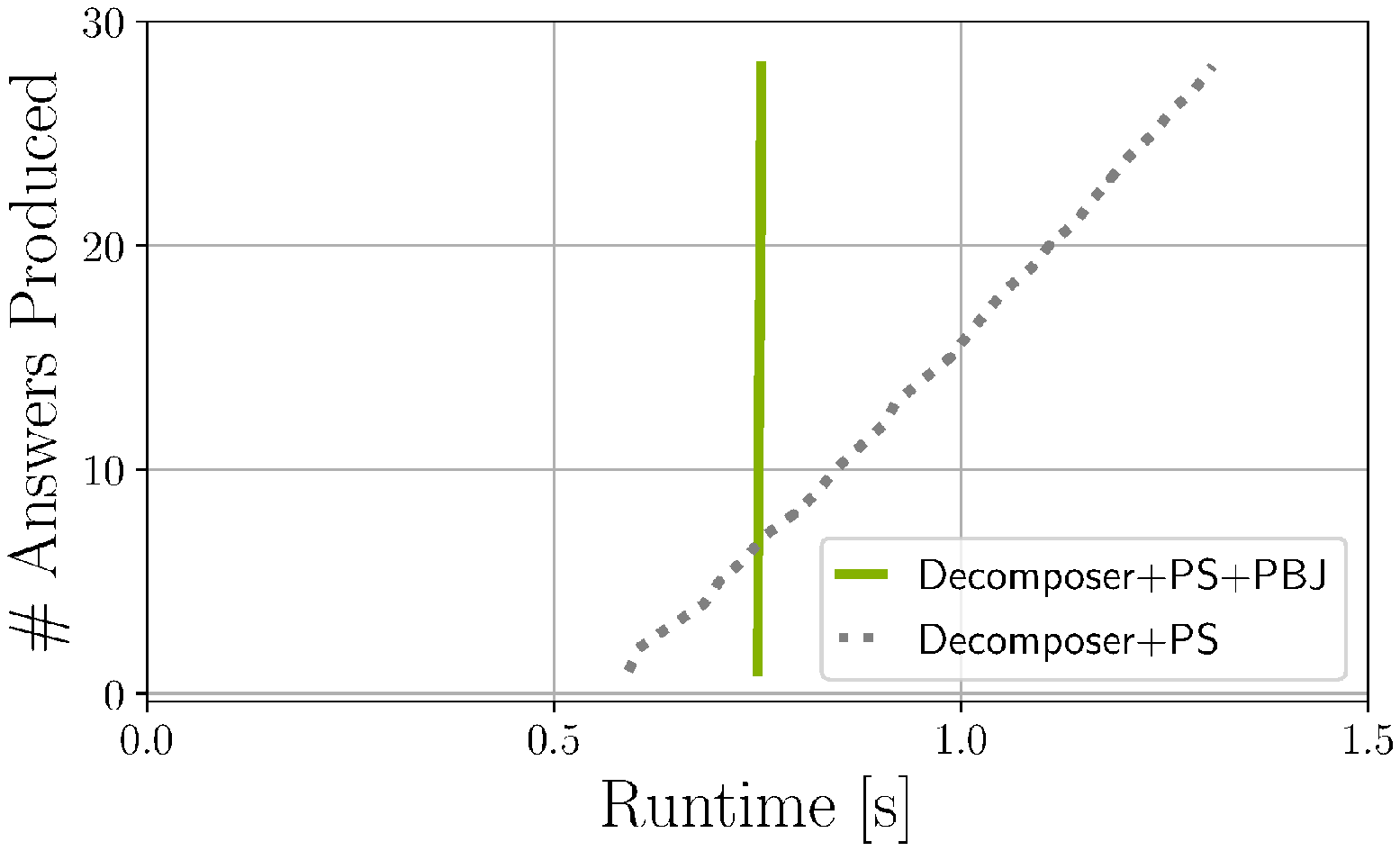}\label{fig:dief_ls6}}
\subfloat[\fedii: \textsf{LS8}]{\includegraphics[width=0.23\textwidth]{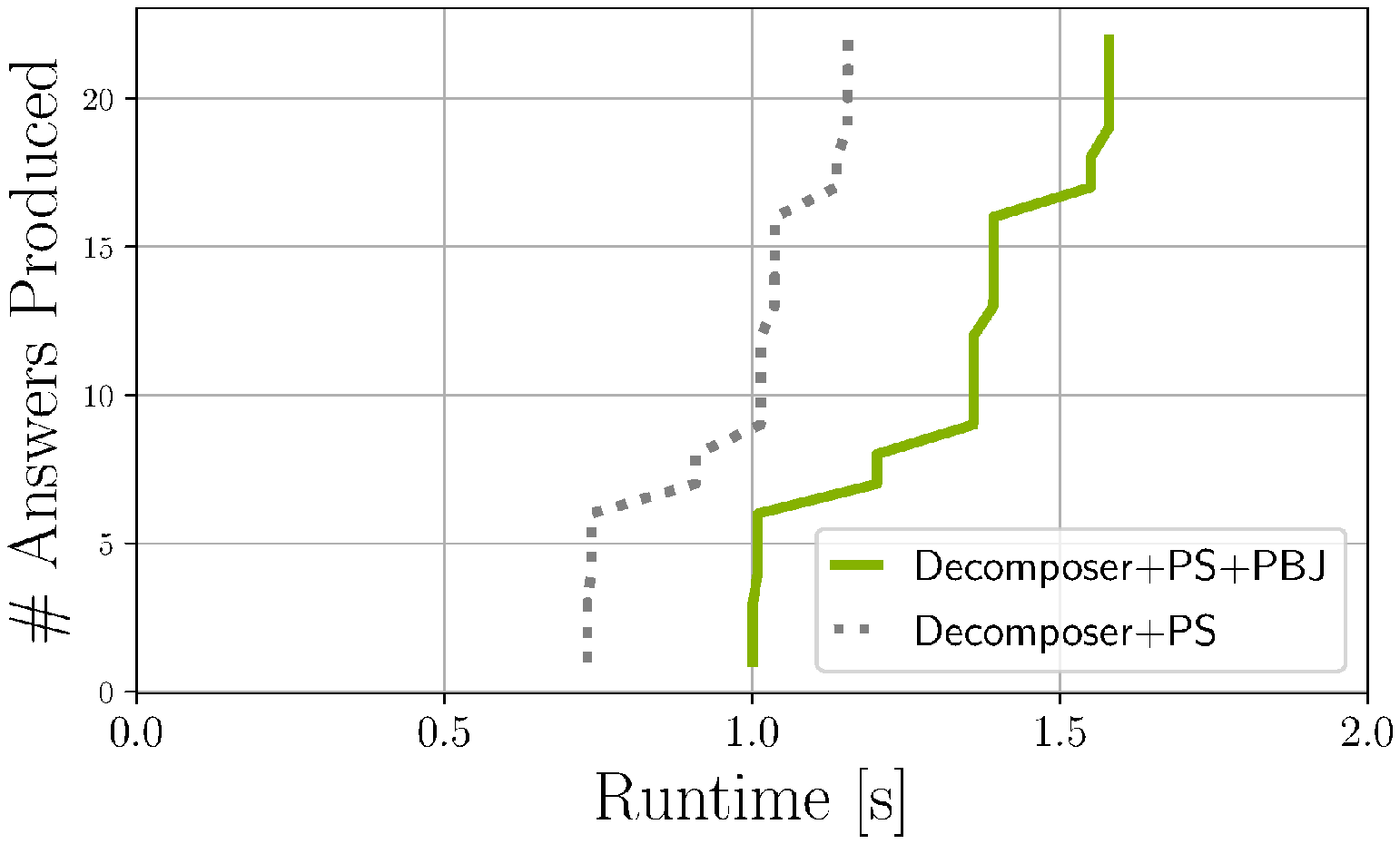}\label{fig:dief_ld8}}
\vspace{-0.7\baselineskip}
\caption{Example diefficiency plots for the approach with the PBJ (green) and without the PBJ (dotted).}
\label{fig:dieffiency}
\end{figure*}

% Complete overview
% --> Table: Runtime, Requests, Answers, 

% Q3 : How is the dieffiency impacted by the PBJ strategy? Full without PBJ / Full with PBJ
%(a) Same performance if the PBJ cannot leverage the capabilities (all TPF)--> same number of requests: 4698 (b) PBJ improves performance when it can leverage the capabilities and thus it significantly reduces the requests (684 vs 1778) (c) PBJ improves performance, because it reduces the requests (53 vs 196) and all PBJs are at the last joins in the plan, and therefore the potential blocking nature of the operators does not reduce performance (d) The blocking nature of PBJ leads to a decrease in diefficiency because it is placed at a join deeper in the query plan. even though it reduces the number of requests (220 vs 266), runtime and dief are reduced.
%After investigating the impact of the decomposer and the source pruning approach, we focus on the proposed polymorphic bind join (PBJ) operator.
\smallbreak
\noindent\textbf{Polymorphic Bind Join.}
%We investigate the proposed polymorphic bind join (PBJ) operator in more detail.
The results in both \cref{fig:runtime_requests} and \cref{tab:evaluation_table} reveal that, in the two federations, adding the Polymorphic Bind Join (\textsf{Decomposer+PS+PBJ}) reduces the number of requests by more than 25\% and, as a consequence, reduces the overall runtimes. 
We investigate the diefficiency to better understand the impact of the PBJ.
In \cref{fig:dieffiency}, we show the diefficiency plots for four example queries.
The plot for \textsf{LS3} in \cref{fig:dief_ls3} shows that the performance of the PBJ is similar to a regular NLJ in case it cannot leverage the capabilities of the interfaces, e.g., 
\textsf{LS3} where only TPFs are contacted.
However, if the capabilities of the services can be leveraged, the PBJ allows for producing the results at a higher rate as shown for queries \textsf{LD3} (\cref{fig:dief_ld3}) and \textsf{LS6} (\cref{fig:dief_ls6}).
In the latter, all answers are produced at once.
Nonetheless, the semi-blocking nature of the PBJ can also have a detrimental effect on diefficiency and runtime as observed for query \textsf{LS8} in \cref{fig:dief_ld8}.
Here, the production of the answers is delayed because the block size of the inner relation (which consumes data from a brTPF server) is not reached until all results of the outer relation in the PBJ are produced. 
Future work could study an adaptive PBJ with variable block sizes, determined according to the expected number of tuples of the outer relation.

Summarizing our experimental evaluation, the results show the effectiveness of our interface-aware techniques for query decomposition, planning, and physical operators.
Furthermore, the results illustrate that our techniques can cope with heterogeneous federations that are composed of different combinations of interfaces.

\smallbreak
\noindent\textbf{Limitations.}
%We evaluated a prototypical implementation of the components of the proposed framework.
The central assumption of our framework is the access to high-level information about the federation (e.g., interfaces and relevant sources) while fine-grained statistics (e.g., data distributions) are not available. 
Therefore, the proposed framework components are limited to devise approximate solutions to the problems of query decomposition, planning, and execution. 
While our experimental results show substantial improvements in the FedBench benchmark, these improvements might not hold in other federations. 
Yet, our framework is a foundation for federated query processing in heterogeneous federations and the components can be refined in the case that additional statistics are available.
For instance, by weighting edges in the decomposition graph according to probability of sources contributing to the answers of a query.

\section{Related Work}
\label{sec:related_work}
% Federated SPARQL Query processing
Query processing over homogeneous federations of SPARQL endpoints has been broadly studied and existing approaches address different challenges.
For instance, \cite{DBLP:conf/semweb/SchwarteHHSS11,DBLP:journals/pvldb/AbdelazizMOAK17} leverage requests during runtime to obtain efficient query plans, while \cite{DBLP:conf/esws/SaleemN14,DBLP:conf/semweb/MontoyaSH17,DBLP:conf/esws/QuilitzL08, DBLP:conf/semweb/GorlitzS11,DBLP:conf/i-semantics/CharalambidisTK15} implement cost models that rely on pre-computed statistics, and \cite{DBLP:conf/semweb/AcostaVLCR11} focuses on runtime adaptivity. 
Furthermore, approaches that specifically study query decomposition have been proposed.
Vidal et al. \cite{DBLP:journals/tlsdkcs/VidalCAMP16} propose a formalization of the query decomposition problem in a way such that it can be mapped to the vertex coloring problem. 
Vidal et al. \cite{DBLP:journals/tlsdkcs/VidalCAMP16}  propose the heuristic \emph{Fed-DSATUR} to solve the problem. 
Similar, Endris et al. \cite{DBLP:conf/dexa/EndrisGLMVA17} formalize the query decomposition problem for federated SPARQL querying and present a decomposition approach that relies on RDF Molecule Templates, which represent metadata obtained by executing SPARQL queries over endpoints. 
Different from our work, these approaches assume all members in the federation to be SPARQL endpoints, and thus, the proposed solutions rely on their querying capabilities.% and would need to be adapted in the presence of less capable interfaces.

% LDF Interfaces
Additional Linked Data Fragment (LDF) interfaces and corresponding SPARQL clients have been proposed.
They range from less expressive interfaces, such as (Bindings-Restricted) Triple Pattern Fragments (\cite{DBLP:journals/corr/HartigA16}) \cite{DBLP:journals/ws/VerborghSHHVMHC16}, to more expressive interfaces such as SaGe \cite{DBLP:conf/www/MinierSM19} or smart-KG \cite{DBLP:conf/www/AzzamFABP20}. 
To study the expressiveness of LDF interfaces, Hartig et al. \cite{DBLP:conf/semweb/HartigLP17} propose the Linked Data Fragment Machines as a formal framework that includes client demand, server demand, and communication cost when executing queries over these interfaces.  
Similar to their work, we also formalize the concept of a server language to distinguish the capabilities of different interfaces in the federation. Yet, our work goes beyond individual interfaces and studies the problem of heterogeneous LDF federations. 
%Further clients to query LDF interfaces have been proposed.

Lastly, a few approaches have addressed the problem of heterogeneous interfaces.  
Comunica \cite{DBLP:conf/semweb/TaelmanHSV18} is a client able to query heterogeneous LDF federations. 
But, in contrast to our work, Comunica does not support interface-aware query decomposition and handles the query execution on a triple pattern level, even if different interfaces are present.
Moreover, the physical join operators, such as the nested loop join, do not adapt to the different interfaces. 
In a recent paper, Cheng and Hartig \cite{cheng2020fedqpl} study query plans in heterogeneous federations. 
Similar to our work, they conceptualize different interfaces, and federation members implementing those interfaces.
They focus on a formal language for logical query plans over such federations but, in contrast to our work, they do not propose specific solutions to devise such plans and derive physical plans to be evaluated by an engine. 
Montoya et al. \cite{DBLP:conf/semweb/MontoyaAH18a} propose a client to query replicas of datasets via heterogeneous interfaces (brTPF server and SPARQL endpoints) to exploit their characteristics.
Different from our work, they focus on different interfaces for single datasets but do no investigate federated querying.

\section{Conclusion and Future Work}
\label{sec:conclusion}
We formalize the concept of federations of Linked Data Fragment services and present the challenges that querying approaches over heterogeneous federations face.
In particular, we present a theoretical framework and practical solutions for query decomposition, query planning, and physical operators tailored to heterogeneous LDF federations.
In our experimental study, we evaluated a prototypical implementation of our proposed solutions.
The results show a substantial improvement in performance achieved by devising interface-aware strategies to exploit the capabilities of TPF, brTPF, and SPARQL endpoints during federated query processing. 
Future work may focus on extending the proposed framework to other LDF interfaces and studying how state-of-the-art query decomposition, planning, and source pruning approaches from federated SPARQL engines can be applied to heterogeneous federations.

\bibliographystyle{ACM-Reference-Format}
\bibliography{acmart.bib}

%\appendix

\end{document}